\documentclass{article}

\usepackage[english]{babel}
\usepackage{amsmath,graphicx}
\usepackage{color}
\usepackage{array}
\usepackage{textcomp}
\usepackage{amsmath, amssymb, amsthm}
\usepackage{mathrsfs}
\usepackage{dsfont}
\usepackage{amsthm}
\usepackage[utf8]{inputenc}
\usepackage{mathabx}
\usepackage{stmaryrd}
\usepackage{csquotes}
\usepackage{anyfontsize}
\usepackage{enumerate}
\usepackage{booktabs}
\usepackage{float}
\usepackage{graphicx}
\usepackage{xcolor}
\usepackage{hyperref}
\usepackage{enumitem}
\usepackage{txfonts}
\usepackage{comment}
\usepackage{algorithm}
\usepackage{lipsum}
\usepackage{mathrsfs} 
\usepackage{bbold}
\usepackage{amsfonts}
\usepackage{graphicx}
\usepackage{epstopdf}
\usepackage{multirow}
\usepackage{array}
\newcolumntype{?}{!{\vrule width 1pt}}
\usepackage{algorithmic}
\usepackage{booktabs}
\ifpdf
  \DeclareGraphicsExtensions{.eps,.pdf,.png,.jpg}
\else
  \DeclareGraphicsExtensions{.eps}
\fi

\usepackage[margin=2.5cm]{geometry}

\usepackage{enumitem}
\setlist[enumerate]{leftmargin=.5in}
\setlist[itemize]{leftmargin=.5in}

\usepackage{amsthm}

\newtheorem{thm}{Theorem}
\newtheorem{prop}{Proposition}

\newtheorem{lem}{Lemma}

\title{Stochastic super-resolution for Gaussian microtextures \thanks{The authors acknowledge the support of the project MISTIC (ANR-19-CE40-005).}}

\author{\'Emile Pierret \thanks{Institut Denis Poisson -- Université d'Orléans, Université de Tours, CNRS
  .(\href{mailto:emile.pierret@univ-orleans.fr}{emile.pierret@univ-orleans.fr},\url{https://www.idpoisson.fr/pierret/},\href{mailto:bruno.galerne@univ-orleans.fr}{bruno.galerne@univ-orleans.fr}, \url{https://www.idpoisson.fr/galerne/})/ }
\and Bruno Galerne \footnotemark[2] \thanks{Institut Universitaire de France (IUF)
  .}}

\usepackage{standalone}
\usepackage{import}

\usepackage{amsopn}

\DeclareMathOperator{\Tr}{Tr}
\DeclareMathOperator{\Var}{Var}
\DeclareMathOperator{\Cov}{Cov}

\DeclareMathOperator{\ADSN}{ADSN}
\DeclareMathOperator{\argmin}{argmin}

\newcommand{\R}{\mathbb{R}}

\newcommand{\Z}{\mathbb{Z}}
\newcommand{\OMN}{{\Omega_{M,N}}}
\newcommand{\OMNr}{{\Omega_{M/r,N/r}}}

\newcommand{\E  }{{\mathbb{E}}}

\newcommand{\I}{{\boldsymbol{I}}}
\newcommand{\Sub}{{\boldsymbol{S}}}
\newcommand{\G}{\boldsymbol{\Gamma}}
\newcommand{\La}{\boldsymbol{\Lambda}}
\newcommand{\A}{{\boldsymbol{A}}}

\newcommand{\B}{{\boldsymbol{B}}}
\newcommand{\M }{{\boldsymbol{M}}}
\newcommand{\bC}{{\boldsymbol{C}}}
\newcommand{\bCt}{\bC_{\bt}}

\newcommand{\T }{{\mathrm{T}}}

\newcommand{\Tx }{\T_{\x}}

\newcommand{\m}{{\boldsymbol{m}}}

\newcommand{\la }{{\boldsymbol{\lambda}}}
\newcommand{\zero}{{\boldsymbol{0}}}
\newcommand{\W}{{\boldsymbol{w}}}
\newcommand{\V  }{{\boldsymbol{v}}}

\newcommand{\U }{{\boldsymbol{u}}}

\newcommand{\USR}{{\U}_\mathrm{SR}}
\newcommand{\XSR}{{\X}_\mathrm{SR}}
\newcommand{\ULR}{{\U}_{\mathrm{LR}}}
\newcommand{\Uref}{{\U}_{\mathrm{ref}}}
\newcommand{\tref}{{\bt}_{\mathrm{ref}}}
\newcommand{\UHR}{{\U}_{\mathrm{HR}}}
\newcommand{\tU}{\tilde{{\U}}}

\newcommand{\X}{{\boldsymbol{X}}}
\newcommand{\tX}{\tilde{{\X}}}

\newcommand{\ka}{{\boldsymbol{\kappa}}}
\newcommand{\bt}{{\boldsymbol{t}}}

\newcommand{\bc}{{\boldsymbol{c}}}
\newcommand{\be}{\boldsymbol{\beta}}

\newcommand{\etab }{\boldsymbol{\eta}}

\newcommand{\balpha }{\boldsymbol{\alpha}}

\newcommand{\z}{{z}}
\newcommand{\x}{{x}}
\newcommand{\y}{{y}}

\newcommand{\e}{\ensuremath{ \mbox{\scriptsize{E}} }}

\usepackage{mathabx}

\usepackage{xcolor}
\definecolor{brickred}{rgb}{0.8, 0.25, 0.33}

\begin{document}
\maketitle
\begin{abstract}
Super-Resolution (SR) is the problem that consists in reconstructing images that have been degraded by a zoom-out operator. This is an ill-posed problem that does not have a unique solution, and numerical approaches rely on a prior on high-resolution images. 
While optimization-based methods are generally deterministic, with the rise of image generative models more and more interest has been given to stochastic SR, that is, sampling among all possible SR images associated with a given low-resolution input.
In this paper, we construct an efficient, stable and provably exact sampler for the stochastic SR of Gaussian microtextures. 
Even though our approach is limited regarding the scope of images it encompasses, our algorithm is competitive with deep learning state-of-the-art methods both in terms of perceptual metric and execution time when applied to microtextures. 
The framework of Gaussian microtextures also allows us to rigorously discuss the limitations of various reconstruction metrics to evaluate the efficiency of SR routines.
\end{abstract}


\setlength{\tabcolsep}{1pt} 
\renewcommand{\arraystretch}{1} 

\section{Introduction}

Super-Resolution (SR) algorithms aim at producing a High Resolution (HR) image corresponding to a Low Resolution (LR) one, the main challenge being to restore sharp edges as well as high frequency texture content that are lost when applying the zoom-out operator.
Since the set of HR images compatible with a given LR input is an affine subspace of high dimension, a strong prior on realistic images is necessary to recover HR images of high visual quality within this subspace. 
In recent works, this prior is generally conveyed either by exploiting a large datasets of HR images via deep learning models or using a reference HR image to specify an adapted model for the unknown HR image.
The former solution is rather generic while the later is particularly adapted for the SR of texture images that benefits from a particular prior on local statistics of the HR unknown image.
Besides, rather than being deterministic, several contributions propose to tackle the one to many dilemma via stochastic SR which consists in sampling among all acceptable HR images associated with an LR input.
Motivated by such recent contributions, in this paper we solve the stochastic SR problem when the HR texture image is assumed to follow a stationary Gaussian distribution.

\subsection{State of the art and related work}

First deep learning approaches for SR propose to optimize the $L_2$ or $L_1$ reconstruction loss using Convolutional Neural Network (CNN) trained in a end-to-end fashion~\cite{Dong_SRCNN_2014_ECCV,Lim_enhanced_network_2017_CVPR_Workshops,Kim_Accurate_image_SR_2016_CVPR,Wei_Sheng_Deep_Laplacian_pyramidal_SR_2017_CVPR}. 
However, the CNN outputs tend to be blurry as they correspond to the mean of plausible SR solutions~\cite{Sonderby_amortised_MAP_2016_ICLR}.
To avoid this regression to the mean issue, other approaches rely on a perceptual loss~\cite{Bruna_Sprechmann_LeCun_super_resolution_with_deep_convolutional_sufficient_statistics_ICLR2016, wang_recovering_realistic_texture_2018_cvpr,Sajjadi_enhancenet_2017_ICCV}, that is, an $L_2$ loss between pre-trained VGG features \cite{Simonyan_Zisserman_very_deep_cnn_vgg_ICLR2015} introduced in~\cite{Johnson_etal_perceptual_losses_for_real-time_style_transfer_and_super-resolution_ECCV2016}, or Generative Advsersial Network (GAN)~\cite{Goodfellow_et_al_2014} conditioned on the LR input image \cite{Ledig_et_al_Photo_realistic_GAN_2017_CVPR, Wang_et_al_2018_ESRGAN}.
More generally SR is a linear inverse problem that can be tackled by most general deep learning techniques for inverse problems in imaging~\cite{Ongie_etal_deep_learning_techniques_for_inverse_problems_in_imaging_2020} that go beyond end-to-end learning, notably variational frameworks such as deep image prior~\cite{Ulyanov_etal_Deep_image_prior_CVPR2018, Ulyanov_etal_Deep_image_prior_IJCV2020} and Plug-and-Play methods~\cite{Kamilov_etal_PnP_methods_for_integrating_physical_and_learned_models_in_computational_imaging_IEEESPM2023}. We do not explore this important body of literature in this work since we are more interested in stochastic methods.

In contrast with the deterministic approaches mentioned above, Stochastic SR algorithms sample within the set of plausible HR images consistent with the LR input.
Lugmayr \emph{et al.} \cite{Lugmayr_et_al_2020_SRFlow_ECCV} use a conditional normalizing flows \cite{Kingma_Dhariwal_glow_NEURIPS_2018} (a network with invertible layers) to learn this distribution via the change-of-variable formula. Diffusion or score-based models are generative networks which gradually learn to transform noise into samples of a distribution data \cite{song_score-based_2023,Ho_DDPM_2020_Neurips}. 
The diffusion framework can be used to train conditional diffusion models specific to SR \cite{Saharia_et_al_SR3_2021_IEEE}.
Alternatively generic diffusion models can be used to provide approximate conditional sampling procedure for inverse problems such as SR \cite{song_score-based_2023,kawar_DDRM_ICLR_2022,Chung_DPS_ICLR_2023,Saharia_et_al_SR3_2021_IEEE, song_pseudoinverse_guided_diff_models_2023_ICLR}, an approach that can be extended to latent diffusion models~\cite{Rout_etal_solving_linear_inverse_problems_latent_diffusion_models_NeurIPS2023}.
Alternatively, Jiang \emph{et al.} \cite{Jiang_CEM_2021_CVF} propose to wrap a module to a classical network to generate stochastic solutions of the SR samples.
In general diffusion models have impressive performance to generate natural images but can have a lack of realism to generate high-frequency textures (as confirmed below by our comparative experiments). 
Other networks can provide results with some artefacts, even though this may be overcome using parametric priors~\cite{chatillon_statistically_2022}.

Another approach is to apply reference-based SR \cite{Freeman_et_al_example_based_SR_2002_IEEE,Freedman_et_al_Image_upscaling_2010_ACM,Zheng_et_al_Crossnet__2018_ECCV,Chang_et_al_SR_embedding_2004_CVPR,Liu_VAE_Reference_based_2021_CVPR}.
In this setting the LR input is provided with a companion HR image that conveys information on the content of the unknown HR solution, e.g. by providing a HR image of the content from another viewpoint.
Particularly relevant for this paper, in the context of super-resolution of stationary texture images, 
the patch distribution of the reference image can be exploited by minimizing the 
the Wasserstein distance between HR output patches and the reference patches \cite{Hertrich_et_al_Wasserstein_patch_prior_superresolution_IEEETCI2022}.
This variational approach, coined Wasserstein Patch Prior (WPP), exploits important statistics proved to be relevant for texture synthesis\cite{Gutierrez_Galerne_Rabin_Hurtut_optimal_patch_assignment_texture_synth_ssvm2017, Galerne_Leclaire_Rabin_semi_discrete_texture_synthesis_SIIMS2018, Houdard_etal_Wasserstein_generative_models_for_patch_based_texture_synthesis_SSVM2021, Houdard_etal_generative_model_for_texture_synthesis_OT_feature_distribution_IJCV2023}. Let us mention that WPP can be used as a training loss for conditional generative models~\cite{Altekruger_Hertrich_WPPNets_and_WPPFlows_SIAMIS2023} and made more robust using semi-unbalanced optimal transport~\cite{Mignon_et_al_Semi_unbalanced_OT_2023_Eusipco, Mignon_et_al_Semi_unbalanced_OT_2024_HAL}.
The WPP framework provides adaptative statistical guarantees to the proven efficiency of patch-based approaches \cite{Zoran_Weiss_from_learning_models_of_natural_image_patches_to_whole_image_restoration_IVVC2011,Yu_Sapiro_Mallat_PLE_GMM_2012}, which is critical for the SR of texture images.

In this work, we will solely focus on a precise class of textures, the so-called Gaussian microtextures~\cite{Galerne_Gousseau_Morel_random_phase_textures_2011}.
This is a simple texture model that assumes that the texture is stationary (statistically invariant by translation) and follows a Gaussian distribution.
Even though the class of texture images well-reproduced by this model is restricted, the Gaussian modeling allows for many probabilistic and statistical tools to be used for processing these textures.
The textures can be summarized by a local texton~\cite{Desolneux_Moisan_Ronsin_texton_icassp2012,Galerne_Leclaire_Moisan_SOTexton_2014}, can be generated on arbitrary large continuous domains~\cite{GLLD_Gabor_noise_by_example_2012, Galerne_Leclaire_Moisan_texton_noise_2017}, extended to video and visually mixed using Wasserstein barycenters between Gaussian distributions \cite{Xia_Ferradans_Peyre_Aujol_synthesizing_mixing_Gaussian_texture_models_2014}.
Closely related to our work, 
the Gaussian modeling allows for solving texture inpainting via Gaussian conditional sampling based on kriging formulas~\cite{Galerne_Leclaire_Moisan_microtexture_inpainting_icassp2016, Galerne_Leclaire_gaussian_inpainting_siims2017}.
This work provides an iterative algorithm based on conjugate gradient descent last work provides an iterative algorithm based on conjugate gradient descent (CGD) that we extend to SR and use as a reference for evaluating our faster ``Gaussian SR'' solution. A preliminary version of the present work has been presented at ICASSP 2023~\cite{Pierret_Galerne_stochastic_SR_for_gaussian_textures_ICASSP2023}.

\subsection{Contributions}

Our approach builds on the kriging framework used for Gaussian inpainting~\cite{Galerne_Leclaire_Moisan_microtexture_inpainting_icassp2016, Galerne_Leclaire_gaussian_inpainting_siims2017}. 
We further exploit the stationarity of the zoom-out operator to obtain an exact and fast convolution-based algorithm for stochastic SR of grayscale Gaussian textures.
Due to correlation between color chanels, this algorithm does not translate straightforwardly to RGB textures, but we propose an approximate algorithm for RGB textures that proves to be experimentally close to the exact but time-consuming algorithm. 
Finally, on a practical viewpoint, we extend our approach to the context of SR given a reference HR image, and show that, when applied to Gaussian textures, our fast Gaussian SR algorithm is both perceptually better and order of magnitudes faster than recent deep learning or optimal transport-based methods.
While our approach has several inherent limitations that we discuss and illustrate, our strong experimental results show that it is of practical interest when dealing with simple microtextures.

\subsection{Plan of the paper}
The plan of the paper is as follows. In \ref{sec:conditionalsrgraycsale} we present our framework and remind results about Gaussian conditional simulation. Then, we detail an exact, stable and efficient simulation procedure of conditional SR for grayscale Gaussian microtextures. 
In~\ref{sec:conditional_sr_rgb}, we propose an approximate algorithm for stochastic SR of RGB micro textures and compare it with the CGD-based reference approach.
The extension of our sampling algorithm in the practical settings of SR given a reference image is discussed in \ref{sec:gaussian_sr_practice} where we also compare our results with state-of-the-art methods. 
This comparaison leads to a discussion on adapted metrics to evaluate textures SR.
Finally, we document the shortcomings of our approach and show that it extends to more general linear operators in \ref{sec:limitations_and_extensions} and conclude our paper in \ref{sec:conclusion}.

\subsection{Notation}

Let $M,N > 2$ be the size of the HR images and $r > 1$ a zoom-out factor such that the integers $M/r,N/r$ are the size of the LR images (it is assumed that $r$ divides both $M$ and $N)$.
We write $[k] = \{0,\ldots,k-1\}$ for any integer $k\geq 1$, $\OMN = [M]\times[N]$ denotes the pixel grid, and all the images are extended on $\Z^2$ by periodization.
For $\U,\V \in \R^{\OMN}$, $\U \star \V$ designates the discrete and periodic convolution defined by \mbox{$
    (\U \star \V)(\x) =  \sum_{\y \in \OMN} \U\big (\x-\y\big)\V(\y)$} for $\x \in \OMN$. Given a kernel $\bt \in \R^{\OMN}$, let us denote $\bCt \in \R^{\OMN \times \OMN}$ the matrix associated with the convolution by $\bt$ such that for $\U \in \R^{\OMN}$, $\bCt \U = \bt \star \U \in \R^{\OMN}$. 
    $\I_{\OMN}$ denotes the identity matrix on $\OMN$ while $\mathbf{1}_{\OMN}$ denotes the constant image of $\R^{\OMN}$ with value $1$.
$\widecheck{\U}$ stands for the symmetric image of $\U$ defined by for  $\widecheck{\U}(\x) = \U(-\x)$ for $\x \in  \OMN$. We express by $\mathscr{F}(\U)$ or $\widehat{\U}$ 
the Discrete Fourier Transform (DFT) of $\U$ defined by
$\widehat{\U}(\x) =   \sum_{\y \in \OMN} \U(\y) \exp({-\frac{2i\pi x_1 y_1}{M}})\exp(-\frac{2i\pi x_2 y_2}{N})$ for $\x \in \OMN$. 
Let us recall that $\mathscr{F}$ is invertible with $\mathscr{F}^{-1} = \frac{1}{MN}\overline{\mathscr{F}}$, $\mathscr{F}(\widecheck{\U}) = \overline{\mathscr{F}(\U)}$ and $\mathscr{F}(\U \star \V) = \mathscr{F}(\U)\odot \mathscr{F}(\V)$ where  $\odot$ denotes the componentwise product. We denote by $\mathscr{F}_r$ the Fourier transform of images of $\R^{\OMNr}$.

\section{The conditional super-resolution of grayscale Gaussian microtextures}
\label{sec:conditionalsrgraycsale}

In this section we solve the problem of Stochastic super-resolution of grayscale Gaussian microtextures using a known zoom-out operator given by a convolution.
Exploiting both the stationarity of the LR measurements and the HR texture model, our main result states that one can exactly sample new HR microtextures by applying a specific convolution operator.
We first recall the specifics of the zoom-out operator and the ADSN microtexture model. Then we recall how solving for kriging coefficients allows for conditional Gaussian simulation and establish our main result.

\subsection{The zoom-out operator}
The most common zoom-out operator used in the SR literature is the Matlab bicubic zoom-out operator provided by the function \texttt{imresize} \cite{Ledig_et_al_Photo_realistic_GAN_2017_CVPR, Wang_et_al_2018_ESRGAN, bahat2020explorable}.
Let us denote $\A$ this bicubic zoom-out operator by a factor $r$. 
It can be described as a convolution followed by a subsampling $\Sub$ with stride $r$ such that for $\U \in \R^{\OMN}$, $\Sub \U \in \R^{\OMNr}$ and for $\x \in \OMNr$, $(\Sub \U)(\x) = \U(r\x)$ and for $\x \in \OMNr, \V \in \R^{\OMNr}$, $(\Sub^T \V)(r \x) = \V(\x)$, see e.g.~\cite{bahat2020explorable}.
We denote by $\bc$ the kernel associated with the bicubic kernel $\A = \Sub \bC_{\bc}$. In the remaining of the paper, we study the following inverse problem
\begin{equation}
	\label{eq:inverse_problem}
	\ULR = \A\UHR
\end{equation}
	where $\ULR$ is the observed low resolution version of an unknown HR image $\UHR$. This problem is ill-posed, its space of solutions is a high-dimensional affine subspace, and most such solutions are not desirable such that blurry HR versions of the LR image $\ULR$. 
	Our main result provides an efficient sampler of solutions of Equation~\ref{eq:inverse_problem} under the assumption that $\UHR$ follows an ADSN distribution, that is, a Gaussian microtexture model we recall below.

\subsection{The ADSN model associated with a grayscale image}
Given a grayscale image $\U\in\R^{\OMN}$ with mean grayscale $m\in\R$,
one defines the Asymptotic Discrete Spot Noise (ADSN) associated with $\U$, $\ADSN(\U)$ as the distribution 
of $m\mathbf{1}_{\OMN} + \bt \star \W$
where $\bt = \frac{1}{\sqrt{MN}}(\U-m\mathbf{1}_{\OMN})$ is called the texton associated with $\U$ and $\W \sim \mathscr{N}(\zero,\I_{\OMN})$ \cite{Galerne_Gousseau_Morel_random_phase_textures_2011}.
Let us recall that in practice it is preferable to replace $\U$ by its periodic component \cite{moisan_periodic_2011} to avoid boundary issues \cite{Galerne_Gousseau_Morel_random_phase_textures_2011}. 
$\ADSN(\U)$ is a Gaussian distribution with mean $m\mathbf{1}_{\OMN}$ and covariance matrix $\G = \bCt\bCt^T = \bC_{\bt \star \widecheck{\bt}} \in \R^{\OMN \times \OMN}$. Samples can be generated fastly executing the convolution with the Fast Fourier Transform (FFT). 
$\ADSN(\U)$ is stationary, that is, for $\X \sim \ADSN(\U)$ and any $\y \in \Z^2$, $\X(. - \y)$ has the same distribution as $\X$. This implies that the pixelwise variance of this distribution is constant.
This restricts the application of ADSN model to microtextures that do not have structures. 
Practically, substracting and adding the common mean gray-level to LR and HR images is considered as a preprocessing and a post-processing. In what follows, we will  often consider the centered distribution $\bt \star \W \sim \mathscr{N}(\zero,\G)$ that will be called $\ADSN(\U)$ with a slight abuse of notation.

\subsection{The kriging reasoning}

Let $\X \in \R^{\OMN}$ be a random vector following a Gaussian multivariate law $\mathscr{N}(\zero,\G)$ with $\G \in \R^{\OMN \times \OMN}$ and let $\A \in \R^{\OMNr \times \OMN}$ be a linear operator. 
The simulation of $\X$ knowing $\A\X$ can be computed using the following theorem, recalled in \cite{Galerne_Leclaire_gaussian_inpainting_siims2017}, which expresses the link between orthogonality and independence for Gaussian vectors.
\begin{thm}[Contidional Gaussian simulation]
\label{thm:Gaussian_conditional_simulation}
Let $\X \sim \mathscr{N}(\zero,\G)$ and $\A$ be a linear operator. The two Gaussian vectors
$\E(\X|\A\X)$ and $\X - \E(\X|\A\X)$ are independent. Consequently, if $\tX$ is independent of $\X$ with the same distribution then $\E(\X|\A\X) + [\tX - \E(\tX|\A\tX) ]$ has the same distribution as $\X$ knowing $\A\X$.
\end{thm}

Furthermore, in the Gaussian context, if the distribution is zero-mean, the relation between the conditional expectation $\E(\X|\A\X)$ and $\A\X$ is linear, as detailed in the following.

\begin{thm}[Gaussian kriging]
\label{thm:kriging}
Let $\X \sim \mathscr{N}(\zero,\G)$ be a \textbf{Gaussian} vector with \textbf{zero mean} and $\A \in \R^{\OMNr \times \OMN}$ be a linear operator.
There exists $\La \in \R^{{\OMNr}\times{\OMN}}$ such that $\E(\X|\A\X) = \La^T\A\X$ and $\E(\X|\A\X) = \La^T\A\X$ if and only if $\La$ verifies the matrix equation
  \begin{equation}
    \label{eq:kriging_matrix}
    \A\G \A^T\La = \A\G.
  \end{equation}
\end{thm}

Theorem~\ref{thm:kriging} is a standard \cite{lantuejoul_geostatistical_2001}. 
Yet since solving Equation~\ref{eq:kriging_matrix} is central in what follows, we give a detailed proof in \ref{appendix:proof_kriging_equation}.
The matrix $\La$ is called the kriging matrix and Equation~\ref{eq:kriging_matrix} is called the kriging equation.  Note that if $\A$  is invertible, a trivial solution for $\La^T$ is $\A^{-1}$. 
In our context of super-resolution, $\A$ is not invertible and we are given an LR image $\ULR = \A\UHR$ with $\UHR$ unknown. Supposing that $\UHR$ is a realization of a law $\ADSN(\U) = \mathscr{N}(\zero,\G)$, we would like to sample HR images $\USR$ such that $\USR$ follows $\mathscr{N}(\zero,\G)$ conditioned on $\A\USR = \ULR$. 
Using Theorem~\ref{thm:Gaussian_conditional_simulation} and Theorem~\ref{thm:kriging}, we aim at computing a HR sample image
\begin{equation}
\label{eq:sample_kriging}
	\USR = \La^T\ULR + \tU - \La^T\A\tU
\end{equation}
 where $\tU \sim \mathscr{N}(\zero,\G)$ and $\La$ verifies the kriging Equation~\ref{eq:kriging_matrix}. 
 Note that $\USR$ is the sum of a deterministic term $\La^T\ULR$,  called the \textit{kriging component}, and a stochastic term $\tU - \La^T\tU$ not depending on $\ULR$, called the \textit{innovation component}. 
 As presented in Figure~\ref{fig:decomposition_kriging}, the kriging component is a HR blurred version of the LR image with added covariance information while the innovation component provides independently the granular aspect of the microtexture. 
 In general, the kriging equation needs to be solved by an iterative method, even with the stationarity assumption of the Gaussian law \cite{Galerne_Leclaire_gaussian_inpainting_siims2017}.
 Due to the stationarity of the Gaussian law $\mathscr{N}(\zero,\G)$ and the convolution form of $\A$, we show below that it is possible to solve Equation~\ref{eq:kriging_matrix} exactly and fastly in the Fourier domain.

\newlength{\fabricwidth}
\setlength{\fabricwidth}{0.196\textwidth}

\begin{figure}
    \centering
    \scriptsize
\begin{tabular}{@{}ccccc@{}}
LR image & HR image & SR Gaussian sample & Kriging comp. & Innovation comp. \\
$\ULR$ & $\UHR$ & $\USR$ & $\La^T\ULR$ & $\tU-\La^T\A\tU$ \\
\includegraphics[width = \fabricwidth]{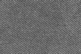}
& \includegraphics[width = \fabricwidth]{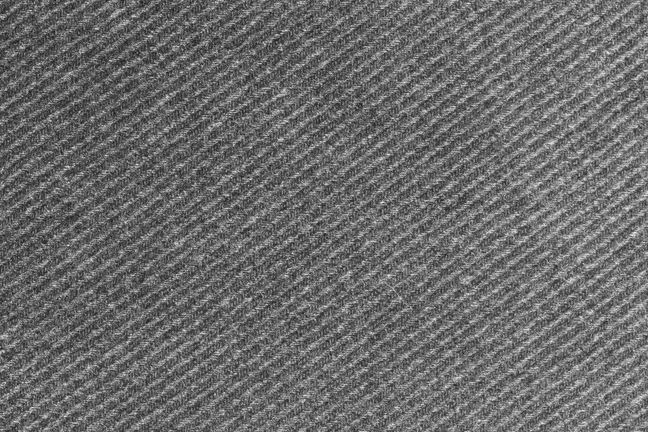}
& \includegraphics[width = \fabricwidth]{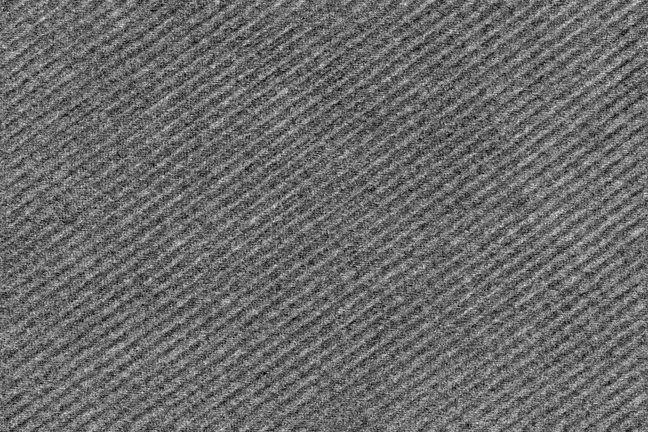}
&\includegraphics[width = \fabricwidth]{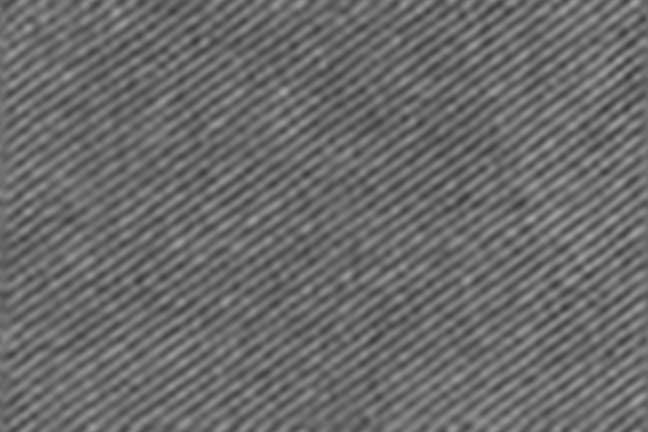}
&\includegraphics[width = \fabricwidth]{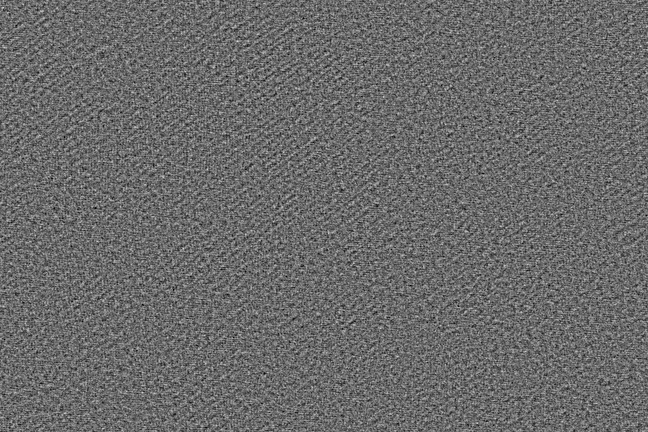}
\end{tabular} 
\caption{\small \label{fig:decomposition_kriging} Super resolution of a Gaussian texture using kriging.
From left to right: LR image, original HR image, conditional Gaussian sample and its corresponding kriging and innovation components.
    HR images size is 416$\times$640 and the zoom-out factor is $r = 8$. 
    The sample is the sum of the kriging component and the innovation component. Grayscale mean is added to all components for visualization. 
    The deterministic kriging component is a kind of pseudo-inversion of the zoom-out operator applied to $\ULR$ that benefits from the knowledge of the covariance of the texture model. 
    The innovation component provides an independent texture grain that is sampled for each realization of $\USR$.
    }
\end{figure}

\subsection{Solving the kriging equation in the super-resolution context}
 
We show in this section that in our SR context, the kriging equation can be solved efficiently with a non-iterative algorithm. 

Equation~\ref{eq:kriging_matrix} may have several solutions and we only consider the specific solution $\La = (\A\G\A^T)^\dagger \A\G$ where $(\A\G\A^T)^\dagger$ is the pseudo-inverse of $\A\G\A^T$. 
The main practical issue is now to compute efficiently $\La^T\V$ for a given LR image $\V \in \R^{\OMNr}$.
As shown by the following proposition, this simply corresponds to applying a specific convolution.

\begin{prop}[Kriging as a convolution]
\label{prop:lambda_convolution}
Let $\A = \Sub \bC_{\bc}$ and let $\G = \bCt\bCt^T = \bC_{\bt \star \widecheck{\bt}}$ be the covariance of a Gaussian distribution $\ADSN(\U)$ with associated texton $\bt = \frac{1}{\sqrt{MN}}(\U-m\mathbf{1}_{\OMN})$.
	$\La = \left(\A\G\A^T\right)^\dagger\A\G$ is an exact solution of Equation~\ref{eq:kriging_matrix} and for all $\V\in \R^{\OMNr}$,
	\begin{equation}
		\La^T\V  = \la\star (\Sub^T \V)
	\end{equation}
	where
		$\la = \bt \star \widecheck{\bt} \star \widecheck{\bc} \star (\Sub^T \ka^\dagger)$
		with $\ka = \Sub(\bt \star \widecheck{\bt} \star \bc \star \widecheck{\bc}) \in \R^{\OMNr}$ and $\ka^\dagger$ the convolution kernel defined in Fourier domain by 
		$$\widehat{\ka}^\dagger(\omega) = \left\{
    \begin{array}{ll}
        \frac{1}{\widehat{\ka}(\omega)} & \mbox{if } \widehat{\ka}(\omega) \neq 0,\\
        0 & \mbox{otherwise},
    \end{array}
	\right.
	\quad \omega \in \R^{\OMNr}.
	$$
\end{prop}

\begin{proof}
Equation~\ref{eq:kriging_matrix} is a normal equation associated with the least squares problem $\argmin_{\La \in \R^{\OMNr \times \OMN}} \left\|\bCt^T\A^T \La -\bCt^T\right\|_2^2$. Consequently, $\La = \left(\A\G\A^T\right)^\dagger\A\G$ is one of the solutions of Equation~\ref{eq:kriging_matrix}.
	Then, we use the following elementary lemma (proved in \ref{appendix:proof_convolution_subsampling} for completeness).
	\begin{lem}[Convolution and subsampling]
\label{lem:convolution_subsampling}
\begin{enumerate}
	\item If $\bC_{\balpha}$ is the convolution on $\OMN$ by the kernel $\balpha \in \OMN$, $\Sub\bC_{\balpha}\Sub^T$ is the convolution on $\OMNr$ by the kernel $\Sub\alpha$, that is, $\Sub\bC_{\balpha}\Sub^T = \bC_{\Sub\balpha}$.
	\item If $\bC_{\be  }$ is a convolution on $\OMNr$ by the kernel $\be \in \OMNr$,   $\Sub^T\bC_{\be} = \bC_{\Sub^T\be}\Sub^T$.
\end{enumerate}
\end{lem}

By~Lemma~\ref{lem:convolution_subsampling}, $\A\G\A^T = \Sub(\bC_{\bc}\G\bC_{\bc}^T)\Sub^T$ is a convolution matrix associated with the kernel $\ka = \Sub(\bt \star \widecheck{\bt} \star \bc \star \widecheck{\bc}) \in \R^{\OMNr}$. Moreover, the pseudo-inverse of a convolution is also a convolution. $\left(\A\G\A^T\right)^\dagger$ is a  convolution associated with the kernel $\ka^\dagger$ such that for $\omega \in \R^{\OMNr}, \widehat{\ka}^\dagger(\omega) = \left\{
    \begin{array}{ll}
        \frac{1}{\widehat{\ka}(\omega)} & \mbox{if } \widehat{\ka}(\omega) \neq 0 \\
        0 & \mbox{otherwise.}
    \end{array}
	\right.$. Also by Lemma~\ref{lem:convolution_subsampling}, $\Sub^T(\A\G\A^T)^\dagger = \Sub^T \bC_{\ka^\dagger} =  \bC_{\Sub^T\ka^\dagger} \Sub^T$.
	Consequently,
	
	$$\La^T = \G^T\A^T\left(\A\G\A^T\right)^\dagger= \G^T\bC_{\bc}^T\Sub^T \bC_{\ka^\dagger}  =  \G^T\bC_{\bc}^T\bC_{\Sub^T\ka^\dagger} \Sub^T = \bC_{\la}\Sub^T$$
	where $\la = \bt \star \widecheck{\bt} \star \widecheck{\bc} \star \Sub^T \ka^\dagger $.
\end{proof}

In what follows we refer to $\la$ as the \textit{kriging kernel}. Figure~\ref{fig:obs_lambda} shows the DFT of the kernel $\la$ associated with the images of Figure~\ref{fig:decomposition_kriging}. 
In contrast with a non-adaptative pseudo-inversion of the bicubic kernel that would be isotropic,
we can observe that the kernel $\la$ is adapted to the covariance structure of the texture and that it amplifies its characteristic frequencies.

To compute SR samples $\USR$ given by Equation~\ref{eq:sample_kriging}, we need to compute $\la \in \R^{\OMN}$ the kernel associated with $\La^T$ as described in Proposition~Proposition~\ref{prop:lambda_convolution} and then apply $\La^T$ as described in Equation~\ref{eq:sample_kriging}.
These two steps can be done fastly in the Fourier domain. 
The exact and fast corresponding procedure is described in Algorithm~\ref{algo:sampling_super_resolution_grayscale}. To generate several samples, one only needs to rerun the second part of the algorithm.

\setlength{\tabcolsep}{1pt} 
\renewcommand{\arraystretch}{1} 

\newlength{\lambwidth}
\setlength{\lambwidth}{0.196\textwidth}
\begin{figure}
\centering
\scriptsize
\begin{tabular}{@{}ccccc@{}}
 HR image & $\la$ & Sample  & Kriging comp. & Innovation comp.  \\
\includegraphics[width=\lambwidth]{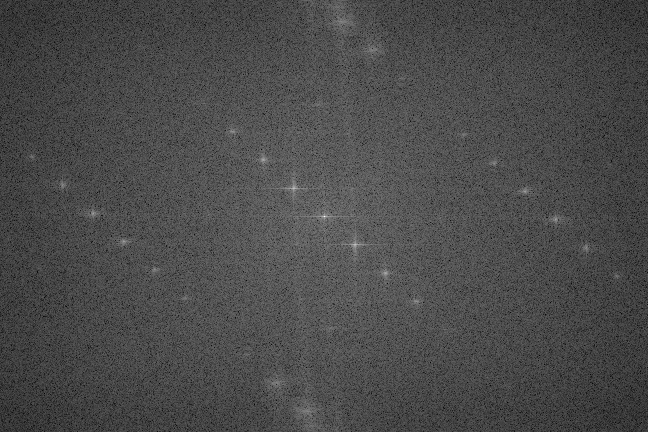} &
 \includegraphics[width=\lambwidth]{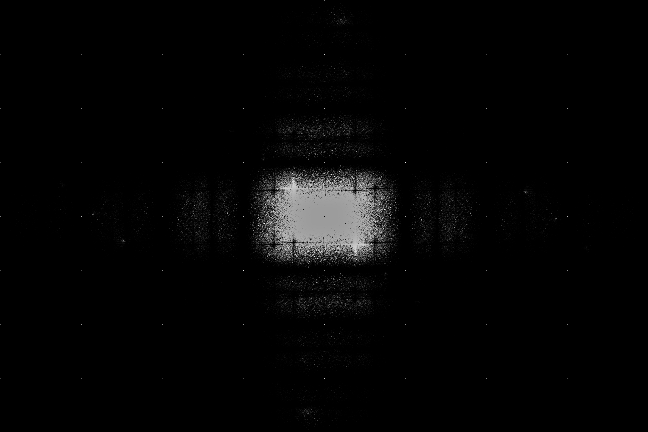} &
  \includegraphics[width=\lambwidth]{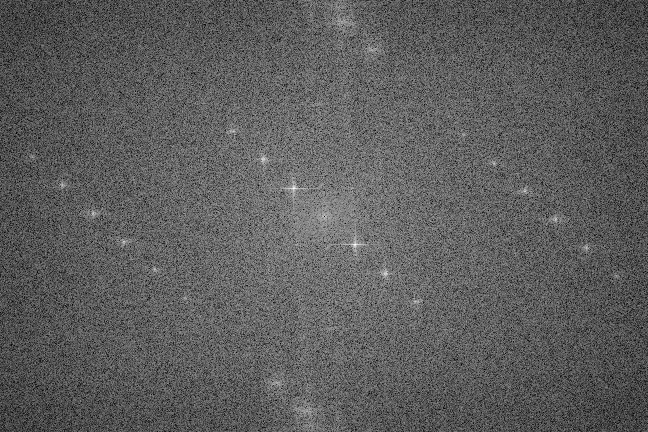} &
  \includegraphics[width=\lambwidth]{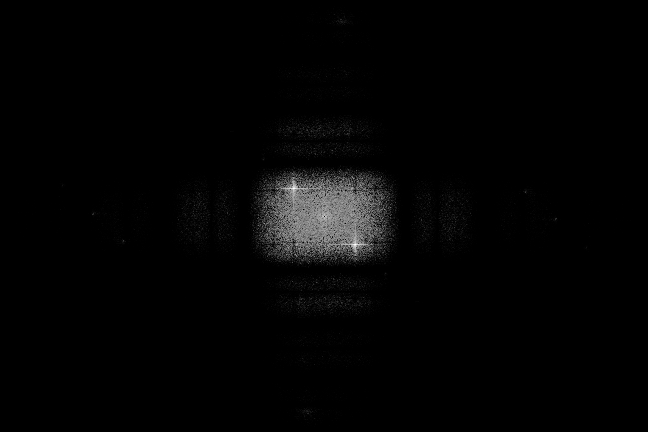} &
    \includegraphics[width=\lambwidth]{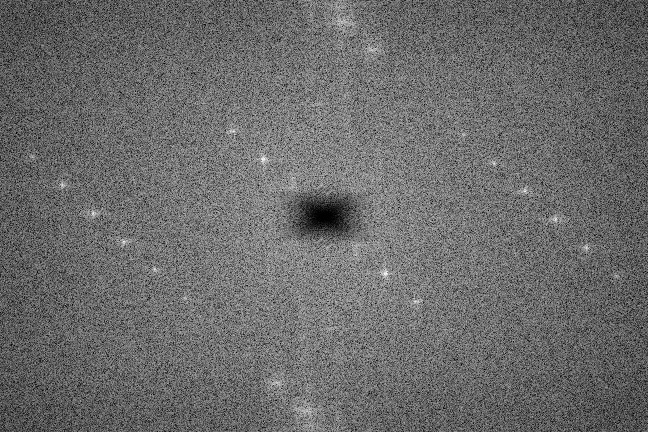} \\
\end{tabular}
	\caption{\small \label{fig:obs_lambda}
Observation of the modulus logscale of DFT of the different components from Figure~\ref{fig:decomposition_kriging}. $\la$ is the kernel associated with the kriging matrix $\La^T$. Note that the kriging component retrieves low frequencies from LR image that are completed with the high frequencies of the innovation component.}
\end{figure}

\begin{algorithm}
\caption{\small Super-resolution sampling for grayscale images with known Gaussian model}
\label{algo:sampling_super_resolution_grayscale}
\begin{algorithmic}
\STATE \textbf{Input:} An image $\ULR \in \R^{\OMNr}$, $r$ the zoom factor, $\bt$ the convolution kernel of the ADSN model, $\bc$ the kernel of the convolution of the zoom-out operator $\A = \Sub\bC_{\bc}$.
\STATE \textbf{Preprocessing:}
\STATE{Compute the grayscale mean $m$ from $\ULR$ and set $\ULR := \ULR - m \mathbf{1}_{\OMN} $}
\STATE{\textbf{Step 1: Computation of the kriging kernel}}
\STATE{Store the DFT transform of the kernel $\la = \bt \star \widecheck{\bt} \star \widecheck{\bc} \star \Sub^T (\ka^\dagger) $}
\STATE{\textbf{Step 2: Simulation of $\USR$}}
\STATE{Sample $\tU = \bt \star \W$ where $\W \sim \mathscr{N}(\zero,\I_{\OMN})$}
\STATE{Compute $\USR =\la  \star \Sub^T(\ULR - \A\tU)+\tU$}
\STATE \textbf{Postprocessing:}
\STATE \textbf{Output:} $ m \mathbf{1}_{\OMN} + \USR$
\end{algorithmic}
\end{algorithm}

The application of $\La^T$ implies the division by $\widehat{\ka}$ in the Fourier domain. However, in practice, this division can be unstable since a non-zero low value of $\widehat{\ka}$ amplifies its corresponding frequency. Actually, in our context, $\La^T$ is applied to LR versions of realizations of an ADSN model that are $\V = \A(\bt \star \W)$ with $\W \in \R^{\OMN}$ an unknown white Gaussian noise. Proposition~\ref{prop:inequality_lambda} ensures that the application of $\La^T$ is stable when applied to images that comply with this assumption. 
\begin{prop}[Stability of the kriging operator on the subspace of the LR ADSN samples]
\label{prop:inequality_lambda}
	Let $\La^T = \G\A^T\left(\A\G\A^T\right)^\dagger \in \R^{\OMN \times \OMNr}$. Then,
	\begin{equation}
	\label{eq:inequation_lambda}
		\forall \W \in \R^\OMN, \left\|\La^T\A(t\star \W)\right\|_2 \leq \|\bCt\|_2  \|\W\|_2 \leq \|\bt\|_1 \|\W\|_2.
	\end{equation}
\end{prop}

\begin{proof}
  As a reminder \cite{penrose_1955}, let $\M$ being a real-valued matrix,
  $$
  	\M^{ \dagger} = \left(\M^T\M\right)^\dagger \M^T
  $$
  With $\M = \A\bCt^{T}$,
  $$
  	\A\bCt^{T \dagger} = \left(\A\bCt(\A\bCt)^T\right)^\dagger\A\bCt = \left(\A\G\A^T\right)^\dagger\A\bCt
  $$
Consequently, for $\W \in \R^{\OMN}$,
$$
\begin{aligned}
	\La^T \A(\bt \star \W)  =  \G\A^T\left(\A\G\A^T\right)^\dagger\A\bCt\W  =  \bCt (\A\bCt)^T\left((\A\bCt)^T\right)^\dagger\W 
\end{aligned}
$$
and
$$
	\left\|\La^T \A(\bt \star \W) 
\right\|_2
 \leq \left\|\bCt\right\|_2 \underbrace{\left\| (\A\bCt)^T\left((\A\bCt)^T\right)^\dagger\right\|_2}_{\leq 1}
 \left\| \W
\right\|_2
 \leq  \left\|\bCt\right\|_2 \left\|\W
\right\|_2
 \leq  \left\|\bt\right\|_1 \left\|\W
\right\|_2,
$$
using that for any matrix $\M$ one has $\|\M \M^{\dagger}\|_2\leq 1$.
\end{proof}

\section{The conditional super-resolution of RGB Gaussian microtextures}
\label{sec:conditional_sr_rgb}

\subsection{Description of the framework and notation}

We denote by $3\OMN$ the space of the RGB images and by $\U_1,\U_2,\U_3$ the three chanels of a given RGB image $\U \in \R^{3\OMN}$. $\widehat{\U}$ designates the Discrete Fourier transform of $\U$ defined as the 2D DFT of each chanel.
The associated texton of an RGB image $\U$ is $\bt = \frac{1}{\sqrt{NM}}(\U-\m) \in \R^{3\OMN}$ where $\m\in \R^{3\OMN}$ is the mean RGB color. The ADSN model $\ADSN(\U)$ is defined as the distribution of $\m + \begin{pmatrix}
						\bt_1 \star \W \\
						\bt_2 \star \W \\
                        \bt_3 \star \W
                    \end{pmatrix}$ where $\W \sim \mathscr{N}(\zero, \I_{\OMN})$.
Note that the white Gaussian noise is the same in each chanel~\cite{Galerne_Gousseau_Morel_random_phase_textures_2011}. 
$\ADSN(\U)$ is a multivariate Gaussian law associated with the covariance matrix $\G = \begin{pmatrix}
   	\bC_{\bt_1} \\
   	\bC_{\bt_2} \\
   	\bC_{\bt_3} \\
   \end{pmatrix}\begin{pmatrix}
   	\bC_{\bt_1} \\
   	\bC_{\bt_2} \\
   	\bC_{\bt_3} \\
   \end{pmatrix}^T \in \R^{3\OMN \times 3\OMN}$ which is the matrix of a multi-channel convolution, that is, each output chanel is a linear combination of 2D convolutions of the input RGB chanels. 
   More precisely, for all $\U \in \R^{3\OMN}$, 
\begin{equation}
	\left(\G\U\right)_i = \bt_i \star \widecheck{\bt_1} \star \U_1+\bt_i \star \widecheck{\bt_2} \star \U_2+\bt_i \star \widecheck{\bt_3} \star \U_3,\quad 1 \leq i \leq 3.
\end{equation}
Unfortunately, this matrix is not diagonal in the Fourier basis which makes a major difference with the grayscale framework.  We still denote by $\A = \Sub\bC_{\bc}$ the zoom-out operator that acts on each chanel for an RGB image.

\subsection{Approximating the kriging operator in the RGB case}

To simulate conditional sample of the RGB ADSN model, the kriging equation \ref{eq:kriging_matrix} and Lemma~\ref{lem:convolution_subsampling} are still valid, except that now $\A\G\A^T$ and $\left(\A\G\A^T\right)^\dagger$ are multi-chanel convolutions. 
However, while the pseudo-inverse of a 2D convolution is diagonal in the Fourier basis, this is not the case for a multi-chanel convolution.

Computing the pseudo-inverse $\left(\A\G\A^T\right)^\dagger$ turns out to be a critical task.
Even for small images, a direct computation of the pseudo-inverse of $\A\G\A^T$ with standard routines leads to instabilities.
A more principled approach is to remark that for each $\V \in \R^{3\OMNr}$ and each $\omega \in \OMNr$, there exists a $3 \times 3$ matrix $\widehat{K}(\omega)$ such that $\mathscr{F}_r(\A\G\A^T \V)(\omega) = \widehat{K}(\omega) \widehat{\V}(\omega)$, reducing the problem to computing one 3$\times$3 pseudo-inverse for each frequency.

Still, each matrix is close to be singular which leads to a high instability in practice.

In order to keep an algorithm which is as fast and as stable as in the grayscale case (Algorithm~\ref{algo:sampling_super_resolution_grayscale}), we propose to make the following approximation: let $\X \in \R^{3\OMN}$, following the Gaussian law $\mathscr{N}(\zero,\G)$,
\begin{equation}
\mathbb{E}\left[\X_i \mid \A\X_1,\A\X_2,\A\X_3\right] \approx \mathbb{E}\left[\X_i \mid \A\X_i\right], \quad 1 \leq i \leq 3.
\end{equation}
 The rational behind this approximation is that to reconstruct the chanel $\X_i$ from the three chanels $\A\X = (\A\X_1 \A\X_2 \A\X_3)$, the more relevant information is in $\A\X_i$. This is a reasonable assumption with respect to the behavior of natural images and the form of the bicubic convolution kernel. 
This approximation amounts to use 
\begin{equation}
\La_{\text{approx}} = \begin{pmatrix}
	\La_1 & \zero & \zero \\
	\zero & \La_2 & \zero \\
	\zero & \zero & \La_3\\
\end{pmatrix} \in \R^{3\OMNr \times 3\OMN}
\end{equation}
as an approximate solution of Equation~\ref{eq:kriging_matrix} in the RGB setting,
where for each $1 \leq i \leq 3$, $\La_i \in \R^{\OMNr \times \OMN}$ is a solution of
\begin{equation}
\A\G_i\A^T\La_i = \A\G_i	, \quad 1 \leq i \leq 3,
\end{equation}
$\G_i$ being the covariance matrix of the $i$th chanel.
While $\La_{\text{approx}}$ is not a solution of the kriging equation, it is the kriging matrix associated with an ADSN law for which the RGB chanels are uncorrelated (which is not an interesting model to generate textures). 
However, we apply this kriging operator $\La_{\text{approx}}^T$ to textures having properly correlated chanels according to the ADSN model $\mathscr{N}(\zero, \G)$.
More precisely, a RGB Gaussian SR sample $\U_{\text{SR}}$ is defined as 
\begin{equation}
	\U_{\text{SR},i} = \La_i^T(\U_{\text{LR},i} - \A\tU_{i}) + \tU_i, \quad 1 \leq i \leq 3,
\end{equation}
where, for $1 \leq i \leq 3$, $\tU_i = \bt_i \star \W$ with $\W \sim \mathscr{N}(\zero,\I_{\OMN})$ the common noise used for each channel and the multiplication by $\La_i^T$ is done following Algorithm~\ref{algo:sampling_super_resolution_grayscale}.

In the next section, we assess that our proposed approximation is harmless by evaluating it with an exact reference iterative alternative.

\subsection{Comparison with the reference CGD algorithm}

In \cite{Galerne_Leclaire_gaussian_inpainting_siims2017}, the authors propose to solve the kriging equation for Gaussian microtextures to resolve the inpainting by conditional simulation. 
To solve Equation~\ref{eq:kriging_matrix}, they apply a CGD to compute $\left(\A\G\A^T\right)^\dagger \varphi$ for a given $\varphi \in \R^{\OMNr}$. For completeness, the algorithm is recalled in \ref{appendix:algo_CGD}. This is an iterative algorithm with a stopping criteria $\varepsilon$ and a given number of iterations. The output $\psi$ of CGD tends to minimize 
\begin{equation}
\label{eq:residual_CGD}
	\left\|(\A\G\A^T)^2\varphi - (\A\G\A^T)\psi\right\|_2
\end{equation} 
which is the residual for the normal equation associated with Equation~\ref{eq:kriging_matrix}.
We consider CGD with a high number of $10^6$ steps as a reference. 
The recommended number of steps for the inpainting problem is $10^2$ for fast results and $10^3$ for high-quality results \cite{Galerne_Leclaire_Gaussian_inpainting_ipol_2017} and our experiments show that these recommendations remain valid for super-resolution. 
Results from our fast convolution-based algorithm and this reference CGD-based algorithm are presented in Figure~\ref{fig:comparison_gradient_Gaussian_SR} for both the grayscale and RGB cases. 
The same realization of noise is used for the simulations so that both algorithm should produce the same image.
Samples are visually similar for the two examples and when visualized as a sequence the texture details added by the CGD algorithm do not evolve after $10^3$ steps.

Ìn~\ref{tab:comparison_gradient_Gaussian_SR}, we compare the samples generated with our algorithm and the CGD routine for different number of steps using the PSNR with respect to the CGD output after $10^6$ iterations and the CGD residual as metrics. For the grayscale image, we observe that our method is indeed exact with very low residual value. 
The high PSNR indicates that our direct Gaussian SR algorithm is more precise than CGD $10^5$ steps.

In the color case, as expected, Gaussian SR is not exact and provides higher residual and PSNR values.
Still the Gaussian SR samples are closer to CGD outputs with a number of steps between $10^4$ and $10^5$, which has been observed for several textures and for different zoom-out factors $r$.
This validates our approach since the minor approximation we introduced is in practice negligible.
Let us recall that the main objective of our approximation is to reduce the simulation time. 
As reported in~\ref{tab:comparison_gradient_Gaussian_SR}, producing such an accurate RGB sample with the CGD algorithm is four to five orders of magnitude longer than with our direct convolution-based algorithm.

\newlength{\cgdwidth}
\setlength{\cgdwidth}{0.196\textwidth}

\begin{figure}
\centering
\tiny
\begin{tabular}{@{}ccccc@{}}
\multicolumn{5}{c}{ Grayscale example } \\
\midrule
LR image & HR image & Gaussian SR (ours) & CGD ($10^2$ steps) & CGD ($10^6$ steps) \\
\includegraphics[width=\cgdwidth]{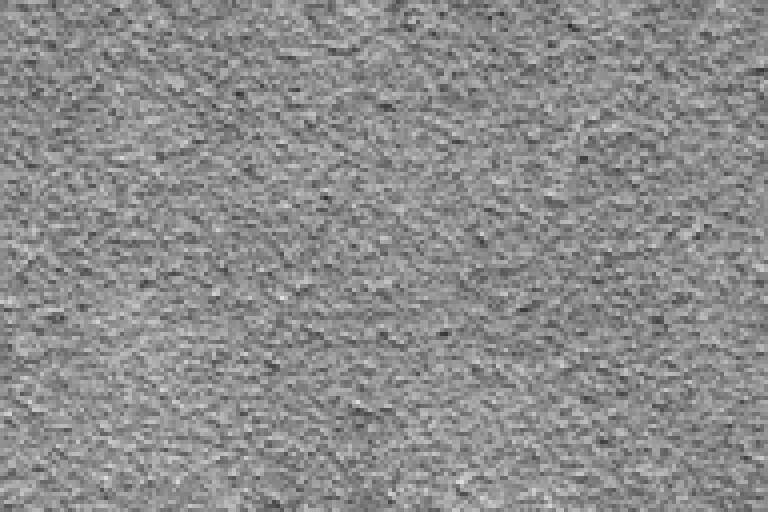}
&
\includegraphics[width=\cgdwidth]{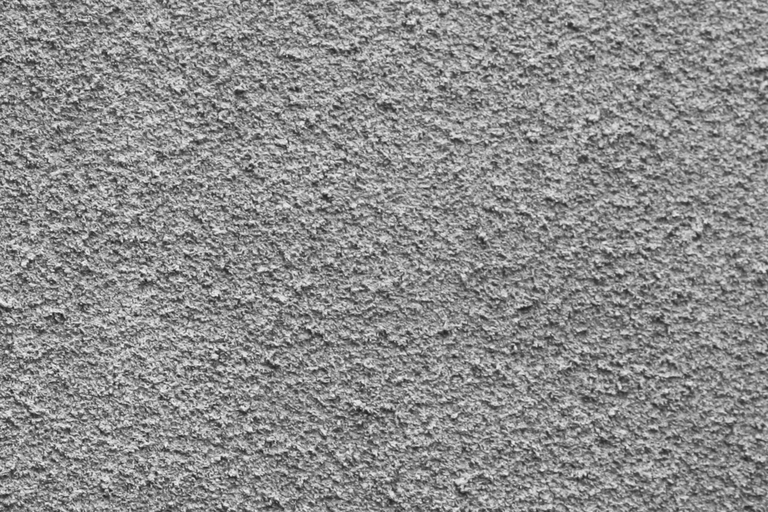}
&
\includegraphics[width=\cgdwidth]{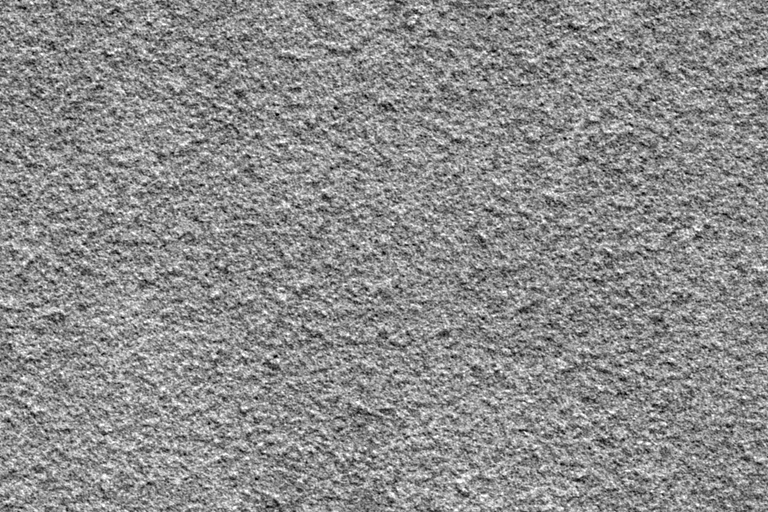}
&
\includegraphics[width=\cgdwidth]{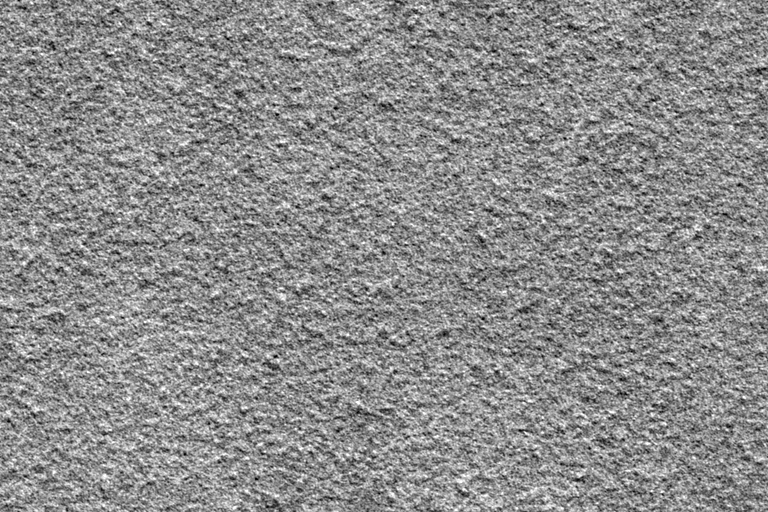} 
&
\includegraphics[width=\cgdwidth]{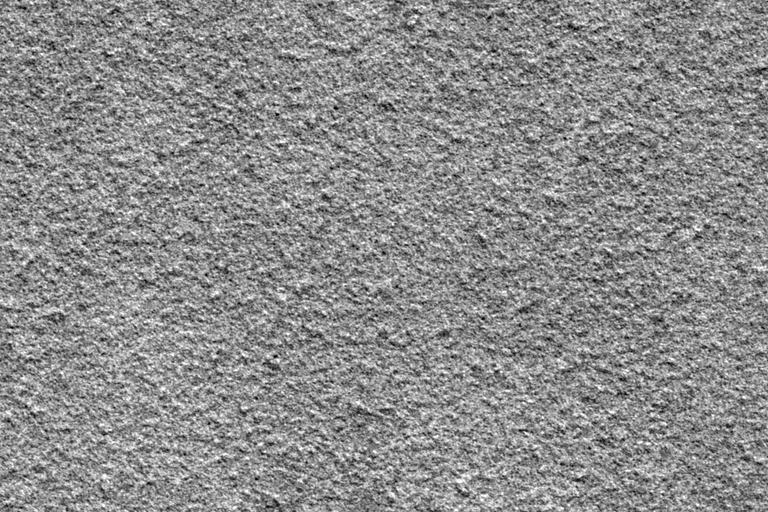}  \\
\multicolumn{5}{c}{ RGB example } \\
\midrule
LR image & HR image & Gaussian SR (ours) & CGD ($10^2$ steps)  & CGD ($10^6$ steps) \\
\includegraphics[width=\cgdwidth]{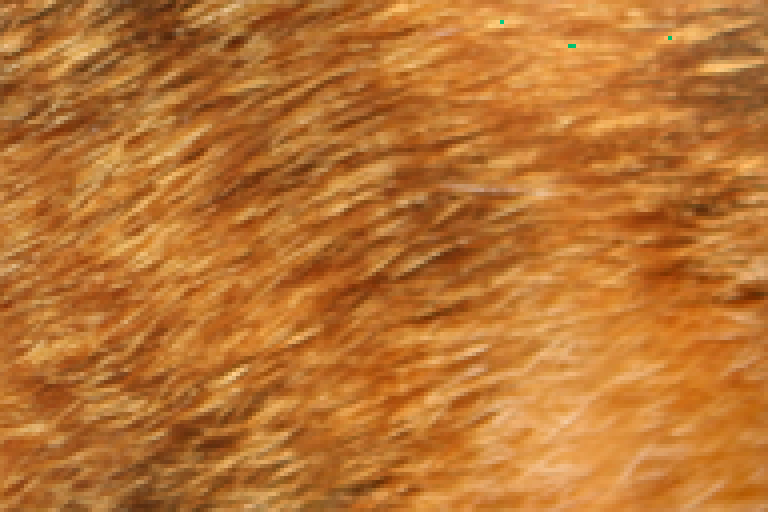}
&
\includegraphics[width=\cgdwidth]{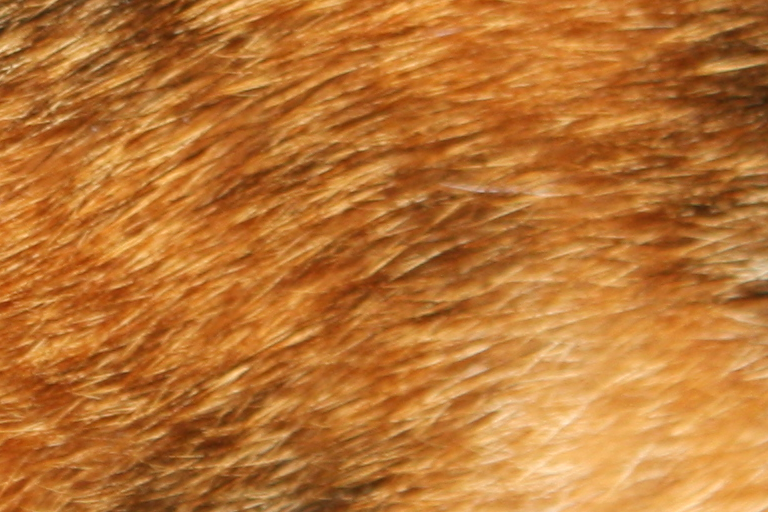}
&
\includegraphics[width=\cgdwidth]{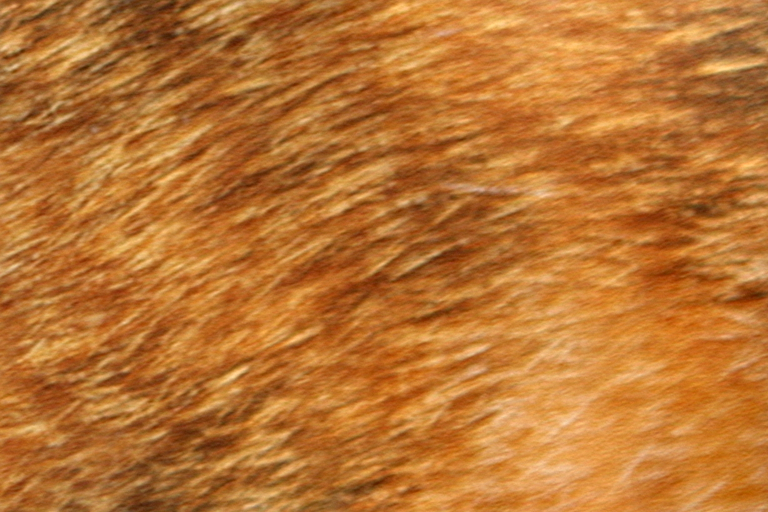}
&
\includegraphics[width=\cgdwidth]{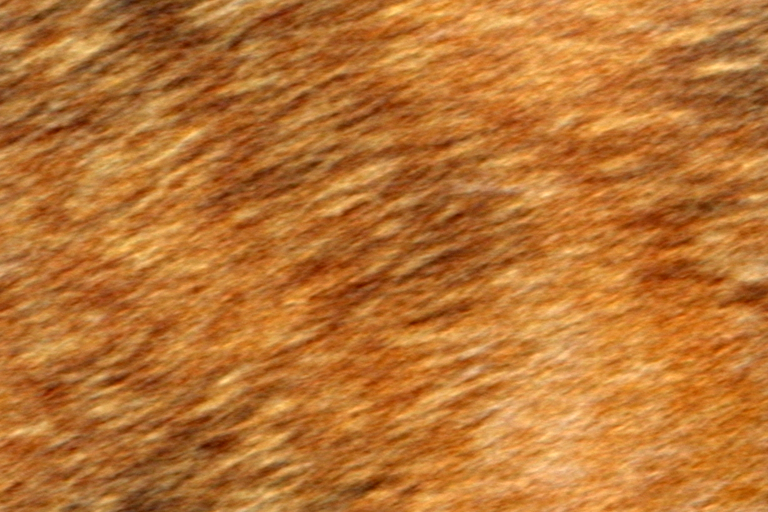} 
&
\includegraphics[width=\cgdwidth]{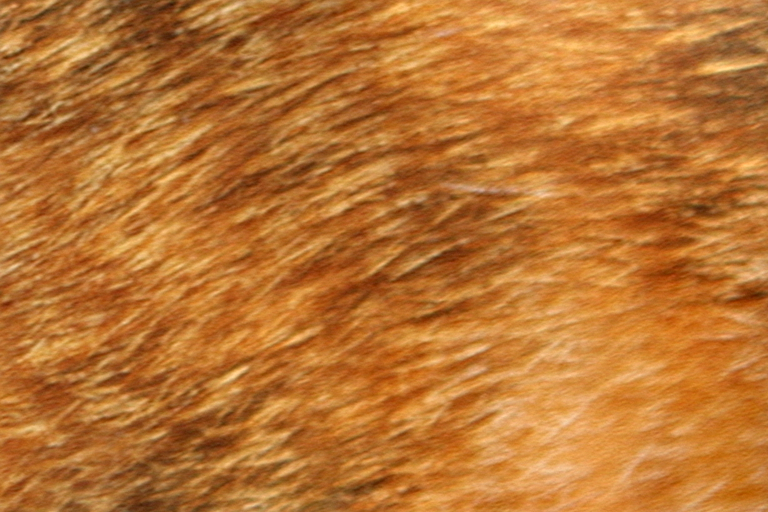}
\end{tabular}
\caption{\small \label{fig:comparison_gradient_Gaussian_SR} Comparison of our direct algorithm (Gaussian SR) and the iterative CGD-based algorithm for grayscale (first row) and RGB (second row) images. 
The HR image size is 512$\times$768, the SR factor is $r = 4$, and the same noise realization is used to compute the innovation components to allow visual comparison. 
The two methods provide similar results even for the RGB case for which our algorithm is not exact.}
\end{figure}

\setlength{\tabcolsep}{6pt} 

\begin{table}
\centering
\footnotesize
\begin{tabular}{llcclcc}
\toprule
\multicolumn{7}{c}{Comparison with the reference CGD algorithm} \\
    \toprule
    & \multicolumn{3}{c}{Grayscale image} & \multicolumn{3}{c}{RGB image} \\
	\cmidrule(lr){2-4}
     \cmidrule(lr){5-7}
     & Residual  & Time(s) &  PSNR w.r.t   &  Residual  & Time(s) &  PSNR w.r.t   \\
     &   & (CPU) & CGD ($10^6$ st.)  &   & (CPU) &  CGD ($10^6$ st.) \\
     \midrule
     Gaussian SR  & $1.32\e{-16}$  &0.01 &   $151.17$ &  $2.54\e{-1}$  & 0.01 & 37.94  \\ 
     CGD ($10^2$ steps)  & $2.74\e{-2}$ & 0.13 & {\color{white}1}29.16 & $2.73\e{0}$   &  0.38   & 25.08\\ 
     CGD ($10^3$ steps)  & $1.03\e{-3}$ & 0.76 & {\color{white}1}47.49 & $3.78\e{-1}$   & 2.49   & 30.19\\ 
     CGD ($10^4$ steps)  & $4.72\e{-5}$ & 7.44 & {\color{white}1}67.60 &  $1.36\e{-1}$  &  23.1 & 35.07 \\ 
     CGD ($10^5$ steps)  & $1.30\e{-8}$ & 81.3 & 145.31 &  $2.37\e{-2}$  &  258 & 39.71 \\ 
     CGD ($10^6$ steps)  & $2.69\e{-42}$ & 749 & - &   $4.97\e{-3}$ &   2588 & -\\ 
\bottomrule 
\end{tabular}
\caption{\small Quantitative evaluation of the results of Figure~\ref{fig:comparison_gradient_Gaussian_SR}.
The residual is the norm from Equation~\ref{eq:residual_CGD}. 
The CGD is executed with different numbers of steps where $10^2$ is the pre-set parameter in the online demo \cite{Galerne_Leclaire_Gaussian_inpainting_ipol_2017} for inpainting. 
The grayscale example confirms that our method is exact in this context. 
The RGB example shows that our approximation is harmless for RGB images.
The different execution times show that producing samples with similar quality to our Gaussian SR samples using the CGD algorithm is four to five orders of magnitude slower.\label{tab:comparison_gradient_Gaussian_SR}}
\end{table}

\setlength{\tabcolsep}{1pt} 

\section{Gaussian SR in practice: Reference image and comparison with state-of-the-art}
\label{sec:gaussian_sr_practice}

\subsection{Gaussian SR with a reference image}
\label{subsec:with_reference_image}

To apply Gaussian SR algorithm, the $\ADSN$ distribution followed by $\UHR$ has to be known (via the computation of the associated texton $\bt$). 
In the experiments of the previous section, the texton is computed from the ground truth HR image $\UHR$,  making the algorithm impractical.

From now on, following recent contributions on texture super-resolution~\cite{Hertrich_et_al_Wasserstein_patch_prior_superresolution_IEEETCI2022, Altekruger_Hertrich_WPPNets_and_WPPFlows_SIAMIS2023},
we consider that the input image $\ULR$ is given with a companion reference HR image $\Uref$ such that both HR images $\Uref$ and the unknown $\UHR$ are supposed to be realizations of the same texture model.
More precisely, in our context based on the ADSN model, Algorithm~\ref{algo:sampling_super_resolution_grayscale} will be called with input an LR image $\ULR$ to be upscaled and the texton  
$\tref = \frac{1}{\sqrt{MN}}(\Uref-m\mathbf{1}_{\OMN})$
associated with the reference image $\Uref$, meaning that $\UHR$ is supposed to be a realization of the $\ADSN(\Uref)$ law.
This makes our routine a stochastic super-resolution algorithm with reference image.

Figure~\ref{fig:gaussian_sr_examples} presents illustrative results that show the practicality of this reference-based approach when the reference images are HR images of another part of the same material.
Figure~\ref{fig:gaussian_sr_examples} also displays the kriging and innovation components and one can observe that the ADSN model from the reference image is good enough to compute an adapted kriging component.

\newlength{\exampleswidth}
\setlength{\exampleswidth}{0.163\textwidth}
\begin{figure}[h!]
\centering
\scriptsize
	\begin{tabular}{@{}cccccc@{}}
		LR image & HR image & Reference image & Sample & Kriging & Innovation \\
		\includegraphics[width = \exampleswidth]{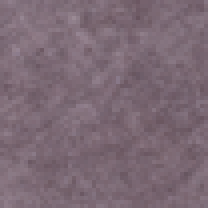}
		& \includegraphics[width =\exampleswidth]{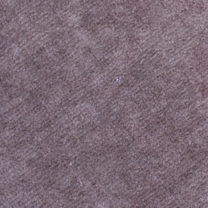}
		& \includegraphics[width = \exampleswidth]{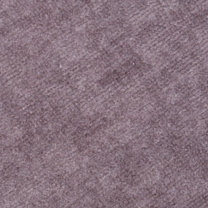}
		& \includegraphics[width = \exampleswidth]{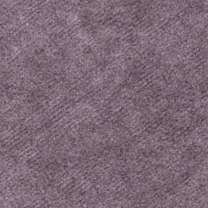}
		& \includegraphics[width = \exampleswidth]{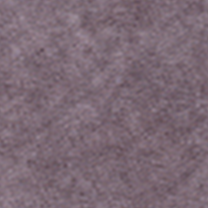}
		& \includegraphics[width = \exampleswidth]{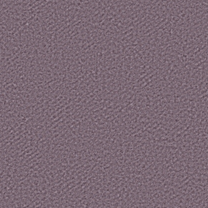} \\
		
		\includegraphics[width = \exampleswidth]{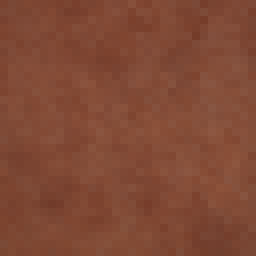}
		& \includegraphics[width =\exampleswidth]{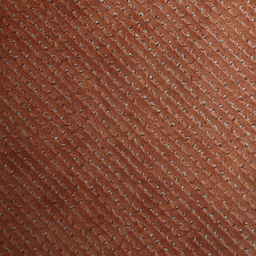}
		& \includegraphics[width = \exampleswidth]{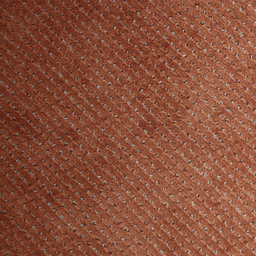}
		& \includegraphics[width = \exampleswidth]{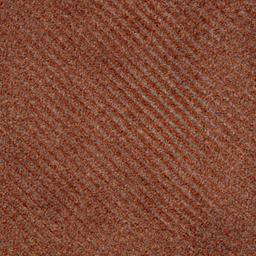}
		& \includegraphics[width = \exampleswidth]{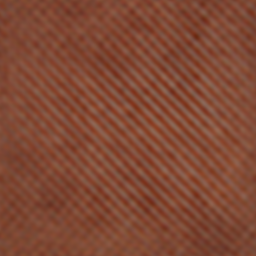}
		& \includegraphics[width = \exampleswidth]{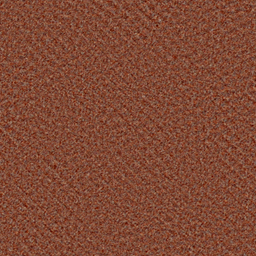} \\
		
		\includegraphics[width = \exampleswidth]{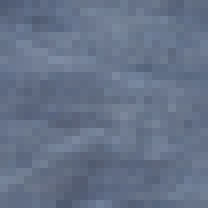}
		& \includegraphics[width =\exampleswidth]{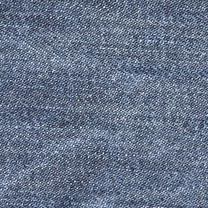}
		& \includegraphics[width = \exampleswidth]{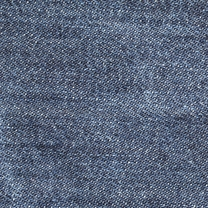}
		& \includegraphics[width = \exampleswidth]{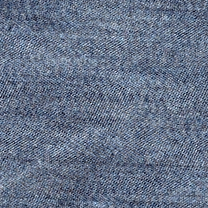}
		& \includegraphics[width = \exampleswidth]{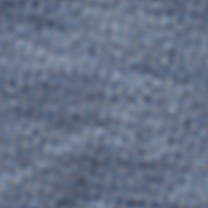}
		& \includegraphics[width = \exampleswidth]{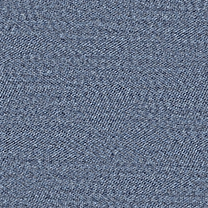} \\
		
		\includegraphics[width = \exampleswidth]{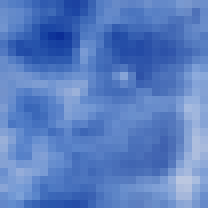}
		& \includegraphics[width =\exampleswidth]{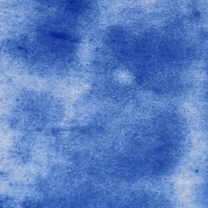}
		& \includegraphics[width = \exampleswidth]{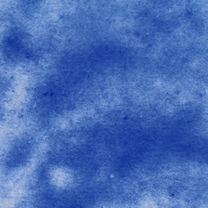}
		& \includegraphics[width = \exampleswidth]{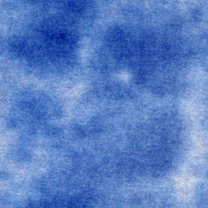}
		& \includegraphics[width = \exampleswidth]{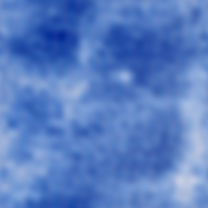}
		& \includegraphics[width = \exampleswidth]{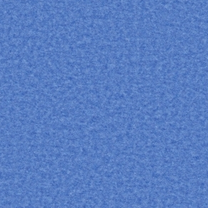} \\
		
		\includegraphics[width = \exampleswidth]{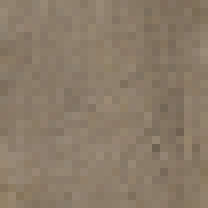}
		& \includegraphics[width =\exampleswidth]{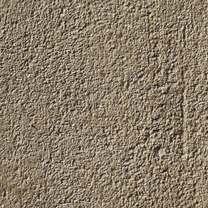}
		& \includegraphics[width = \exampleswidth]{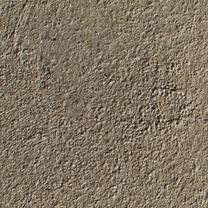}
		& \includegraphics[width = \exampleswidth]{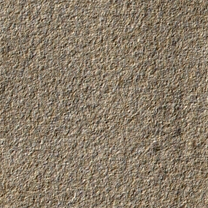}
		& \includegraphics[width = \exampleswidth]{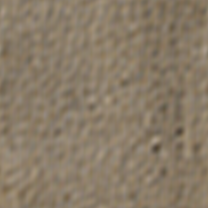}
		& \includegraphics[width = \exampleswidth]{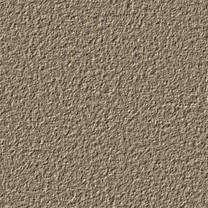} \\
		
     	 \includegraphics[width = \exampleswidth]{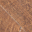}
		& \includegraphics[width =\exampleswidth]{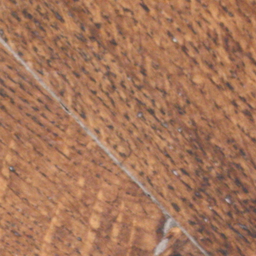}
		& \includegraphics[width = \exampleswidth]{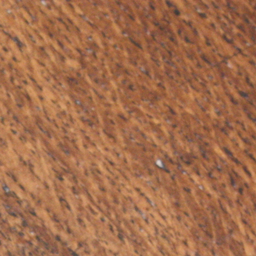}
		& \includegraphics[width = \exampleswidth]{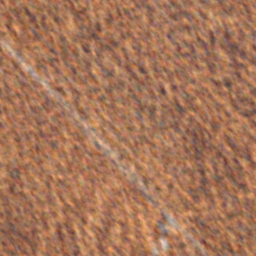}
		& \includegraphics[width = \exampleswidth]{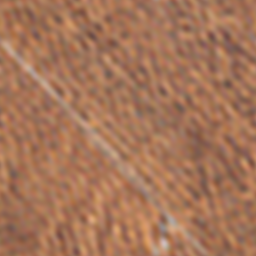}
		& \includegraphics[width = \exampleswidth]{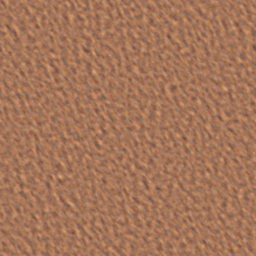} \\
	\end{tabular}
	\caption{\small \label{fig:gaussian_sr_examples}  Illustration of Gaussian SR with a reference image (HR size is 208$\times$208 and $r = 8$). Our method performs well on fabric textures and we can observe the grain provided by the innovation component for different types of texture. Some details are not recovered by our algorithm due to the stationarity assumption.}
\end{figure}

\subsection{Comparison with other methods}

Contrary to neural networks trained on a very large dataset of generic images, our method only works on the restricted area of Gaussian microtextures. This limited scope is a main drawback of Gaussian SR that cannot work on more generic images. Moreover, our method necessitates a HR reference image.
Still it is interesting to compare our method with existing super-resolution algorithms, although one should keep in mind that these methods are designed to be used on a broader range of images.

We compare our method with other super-resolution methods on two microtextures in Figure~\ref{fig:comparaison_methods_Wall},\ref{fig:comparaison_methods_rerA} and Table~\ref{tab:metrics}.
Let us describe the concurrent algorithms and the metrics used for comparison.

Super-Resolution using Normalizing Flow (SRFlow\footnote{ Code and weights from \url{https://github.com/andreas128/SRFlow}}) \cite{Lugmayr_et_al_2020_SRFlow_ECCV}, Denoising Diffusion Restoration Models (DDRM\footnote{ Code and weights from \url{https://github.com/bahjat-kawar/ddrm}}) \cite{kawar_DDRM_ICLR_2022} and  Diffusion Posterior Sampling (DPS\footnote{ Code and weights from \url{hhttps://github.com/DPS2022/diffusion-posterior-sampling}}) \cite{Chung_DPS_ICLR_2023} are diverse super-resolution routines where several samples are proposed to solve the inverse problem \ref{eq:inverse_problem}. SRFlow is a normalizing flow network which has been trained during 5 days on a single NVIDIA V100 GPU with a general dataset of general images DIV2K \cite{div2K_Agustsson_2017_CVPR_Workshops} to solve the SR problem. It depends on a hyperparameter $\tau$, the temperature, that modulates the variance of the latent space. 
DPS and DDRM work with a diffusion model that has been trained on the dataset ImageNet \cite{imagenet_deng_2009} for image generation during several days on GPU \cite{Dhariwal_et_al_Diffusion_beats_gan_Neurips_2021}. A step of attachment to degraded data is added in the generation to solve inverse problems such that SR. DPS proposes a gradient term and DDRM an adaptative reverse process of the score-based model depending on the singular values of the operator $\A$. 
Similar to our method, Wasserstein Patch Prior (WPP\footnote{\url{https://github.com/johertrich/Wasserstein_Patch_Prior}}) \cite{Hertrich_et_al_Wasserstein_patch_prior_superresolution_IEEETCI2022} proposes to solve SR problem in the SISR setting using a reference image.  
It is a multiscale iterative algorithm that optimizes at several resolutions a optimal transport distance between patches of the proposed image and the ones of the reference image. 
In this comparative study, only WPP and Gaussian SR use the reference image while SRFlow, DDRM and DPS are trained on a larger dataset of generic images (and are not restricted to texture restoration).

To compare the methods, we use the four metrics Peak Signal to Noise Ratio (PSNR), LR-PSNR, Structural SIMilarity(SSIM)  and  Learned  Perceptual  Image  Patch  Similarity  (LPIPS). The PSNR is the logarithmic scale of the mean-square error (MSE) between the HR image $\UHR$ and the solution proposed by the method. LR-PSNR is the PSNR between the LR image and the LR version of the solution and quantifies fidelity to the input data $\ULR$. The SSIM~\cite{wang_image_2004} quantifies the similarity between images relying on a luminance, a contrast and a structure terms. LPIPS~\cite{zhang_unreasonable_2018} is a perceptual metric that compute the weighted squared distance between features of the two images of a classification network. The weights have been trained with respect to perceptual criteria.

Quantitatively, as observed in Figure~\ref{fig:comparaison_methods_Wall},\ref{fig:comparaison_methods_rerA}, DDRM tends to produce blurred samples and DPS smooth images and consequently do not really recover the texture. SRFlow with a positive temperature gives good textured solutions but with artefacts generally observed in networks' outputs. WPP generates textured but too smooth image. 
We can also observe that the different luminance of the reference image in Figure~\ref{fig:comparaison_methods_rerA} influences the WPP result. 
Gaussian SR gives the desired perceptual grain of the textures. However, it does not retrieve correctly the details, as it can be observed with the white spot in Figure~\ref{fig:comparaison_methods_Wall}.

\newlength{\srrefwidth}
\setlength{\srrefwidth}{0.163\textwidth}

\begin{figure}
\scriptsize
\centering
	\begin{tabular}{@{}cccccc@{}}
	LR image & HR image & Reference image & Gaussian SR (ours) & Kriging comp. & Innovation comp. \\ 
	\includegraphics[width = \srrefwidth]{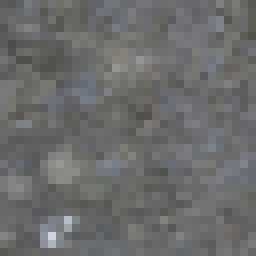} 
	& \includegraphics[width = \srrefwidth]{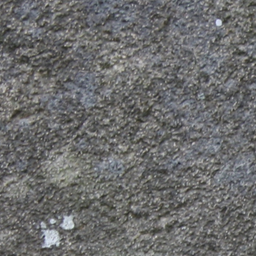} 
	& \includegraphics[width = \srrefwidth]{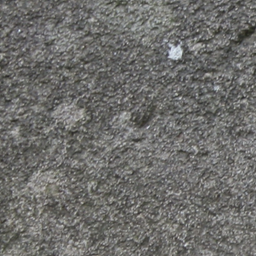} 
	& \includegraphics[width = \srrefwidth]{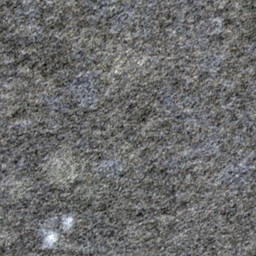} 
	& \includegraphics[width = \srrefwidth]{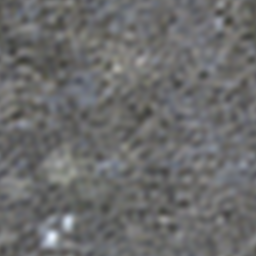}
	& \includegraphics[width = \srrefwidth]{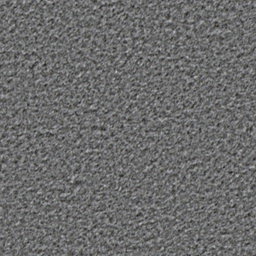}  \\
	Bicubic & WPP & SRFlow ($\tau = 0$) & SRFlow ($\tau = 0.9$) & DDRM & DPS \\  
		\includegraphics[width = \srrefwidth]{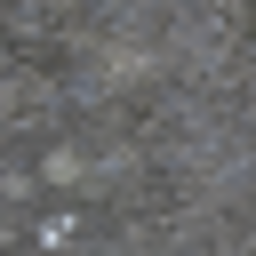} 
	& \includegraphics[width = \srrefwidth]{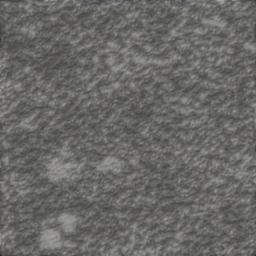} 
	& \includegraphics[width = \srrefwidth]{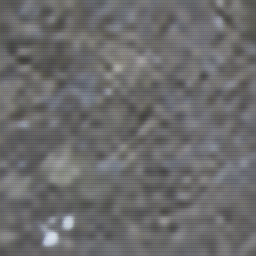} 
	& \includegraphics[width = \srrefwidth]{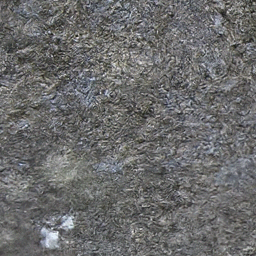} 
	& \includegraphics[width = \srrefwidth]{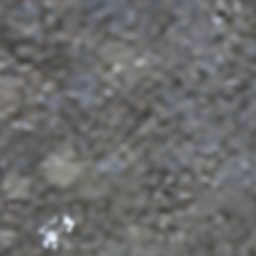}
	& \includegraphics[width = \srrefwidth]{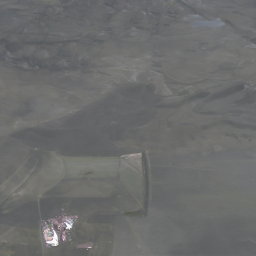}  \\
	\end{tabular}
\caption{\small \label{fig:comparaison_methods_Wall} Comparison of our method with the stochastic SRflow, DPS and DDRM and with the SISR WPP that uses the reference image (The HR size is 256$\times$256 and $r = 8$.). Note that our method is unable to provide the white spot in the wall that does not respect the stationarity assumption. However it faithfully conveys the granular aspect of the texture given by the reference image.}
\end{figure}

\begin{figure}
\scriptsize
\centering
	\begin{tabular}{@{}cccccc@{}}
	LR image & HR image & Reference image & Gaussian SR (ours) & Kriging comp. & Innovation comp. \\ 
	\includegraphics[width = \srrefwidth]{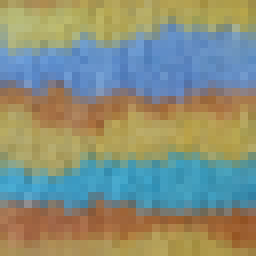} 
	& \includegraphics[width = \srrefwidth]{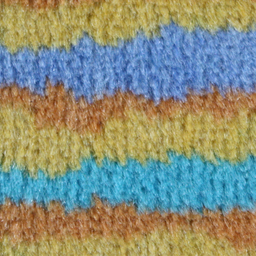} 
	& \includegraphics[width = \srrefwidth]{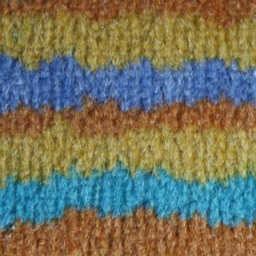} 
	& \includegraphics[width = \srrefwidth]{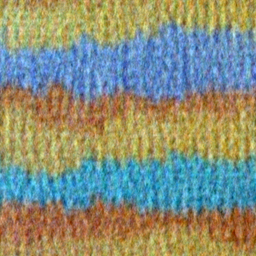} 
	& \includegraphics[width = \srrefwidth]{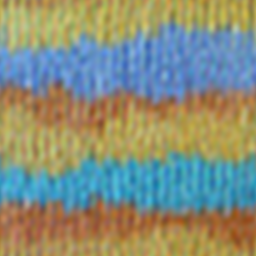}
	& \includegraphics[width = \srrefwidth]{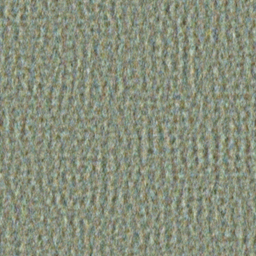}  \\
	Bicubic & WPP & SRFlow ($\tau = 0$) & SRFlow ($\tau = 0.9$) & DDRM & DPS \\  
		\includegraphics[width = \srrefwidth]{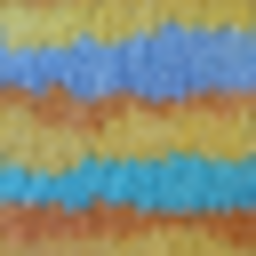} 
	& \includegraphics[width = \srrefwidth]{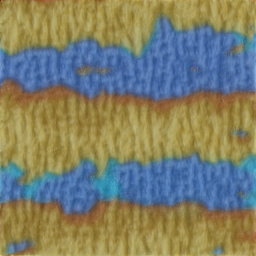} 
	& \includegraphics[width = \srrefwidth]{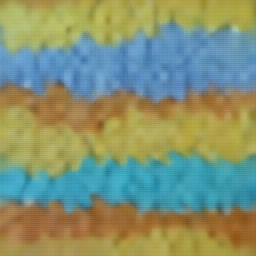} 
	& \includegraphics[width = \srrefwidth]{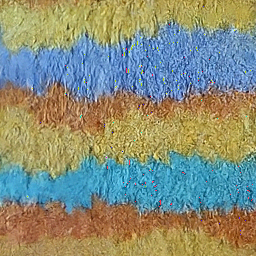} 
	& \includegraphics[width = \srrefwidth]{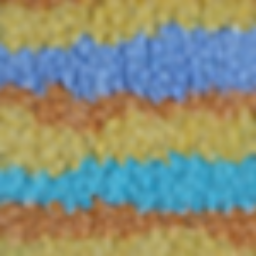}
	& \includegraphics[width = \srrefwidth]{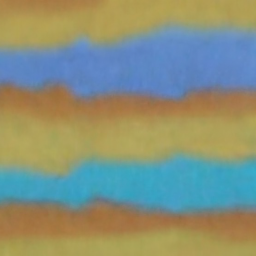}  \\
	\end{tabular}
\caption{\small \label{fig:comparaison_methods_rerA}Comparison of our methods with stochastic SRFlow, DPS and DDRM and with the SISR WPP that uses the reference image (The HR size is 256$\times$256 and $r = 8$.). 
Contrary to WPP, Gaussian SR is not influenced by the difference of luminance between the HR and the reference images.}
\end{figure}

\begin{table}
	\centering 
	{\fontsize{6.5pt}{7pt}\selectfont 
	\begin{tabular}{lllllllllll}
	\toprule
     & \multicolumn{5}{c}{Figure~\ref{fig:comparaison_methods_Wall}} & \multicolumn{5}{c}{Figure~\ref{fig:comparaison_methods_rerA}} \\
     \cmidrule(lr){2-6}
     \cmidrule(lr){7-11}
    &{\tiny PSNR $\uparrow$} & {\tiny LR-PSNR $\uparrow$}  & {\tiny SSIM $\uparrow$}  & {\tiny LPIPS $\downarrow$} & {\tiny Time} & {\tiny PSNR $\uparrow$}  & {\tiny LR-PSNR $\uparrow$} & {\tiny SSIM $\uparrow$}  & {\tiny LPIPS $\downarrow$ } &{\tiny Time}\\
    \midrule
    Gaussian SR (ours) & 17.05$\pm$0.04 &  \underline{159.24$\pm$0.04} & 0.08$\pm$0.01 & $\mathbf{0.22}$$\mathbf{\pm}$$\mathbf{0.01}$ & \textbf{0.01}$^1$  & 19.14$\pm$0.09 & $\mathbf{154.52}$$\mathbf{\pm}$$\mathbf{0.36}$ & 0.20$\pm$0.01 & \textbf{0.23$\pm$0.01} & \textbf{0.01}$^1$ \\ 
\hline
 Kriging comp. & 18.76 & $\mathbf{159.30}$ & 0.11 & 0.75 & - & 21.42 &  \underline{154.47}  & \textbf{0.30} & 0.52 & -\\
\hline
 Bicubic & 20.74 & 37.61 & \underline{0.18} & 0.76 &  -  & 20.74 &  38.35 & 0.28 & 0.66 & - \\ 
\hline
WPP & 20.04 & 29.29 & 0.14 & 0.36 & 44.0$^2$ & 17.68 & 18.84 & 0.19 & \underline{0.30} &  64.0$^2$ \\ 
\hline
SRFlow {\tiny ($\tau=0$) }& \underline{21.44} & 54.61 & \textbf{0.20} & 0.70 & 0.22$^2$ & \underline{21.64} &  51.63  & \underline{0.29} & 0.54 & 0.23$^2$ \\ 
\hline
SRFlow {\tiny ($\tau=0.9$)} & 18.29$\pm$0.36 & 55.13$\pm$0.15 & 0.12$\pm$0.01 & \underline{0.30$\pm$0.03} & \underline{0.19}$^2$   & 18.21$\pm$0.53 & 54.02$\pm$0.23 & 0.16$\pm$0.02 & 0.39$\pm$0.06 & \underline{0.20}$^2$ \\  
\hline
DDRM& $\mathbf{21.57}$$\mathbf{\pm}$$\mathbf{0.05}$ & 56.59$\pm$0.17 & $\mathbf{0.20}$$\mathbf{\pm}$$\mathbf{0.00}$ & 0.70$\pm$0.02 &  1.66$^2$   & $\mathbf{22.44}$$\mathbf{\pm}$$\mathbf{0.04}$ & 56.02$\pm$0.20 & $\mathbf{0.30}$$\mathbf{\pm}$$\mathbf{0.00}$ & 0.55$\pm$0.01 & 1.68$^2$ \\ 
\hline
DPS & 20.59$\pm$0.22 &  {\color{white}5}7.61$\pm$0.02 & 0.14$\pm$0.01 & 0.65$\pm$0.08 & 132$^2${\color{white}.} & 21.36$\pm$0.06 &  {\color{white}5}9.39$\pm$0.02 & 0.25$\pm$0.00 & 0.66$\pm$0.03 & 145$^2${\color{white}.} \\
\bottomrule
\multicolumn{8}{l}{\footnotesize $^1$CPU,$^2$GPU Nvidia A100}
\end{tabular} }
	\caption{\small \label{tab:metrics} Quantitative comparison with the state-of-the-art methods results from Figure~\ref{fig:comparaison_methods_Wall},\ref{fig:comparaison_methods_rerA} realized on $100$ samples for stochastic methods. Blurry results outperfoms on SSIM and PSNR metrics while our method provides the best perceptual LPIPS metric. Our method provides also a strong attachment to data, illustrated by the LR-PSNR metric. Our method is also faster, working on CPU.}
\end{table}

\ref{tab:metrics} shows that our method is faster and have the best LPIPS on each image, followed by WPP. Note that Gaussian SR works on CPU. Our method proposes also a strong data attachment, illustrated by the high LR-PSNR metric. It is well-known that PSNR promotes blurry images, a limitation called the ``regression to the mean'' problem \cite{Sonderby_amortised_MAP_2016_ICLR} in the SR litterature. Our image with the best PSNR comparison is the kriging component. We can explain it theoretically by the next Proposition~\ref{prop:MSE} which shows that in expectation the best PSNR is provided by the mean of our samples, the kriging component.

\begin{prop}[Kriging component and MSE]
\label{prop:MSE}
Let $\UHR \in \R^{\OMN}$ be a HR image, $\ULR = \A\UHR$ its LR version, $\La \in \R^{\OMNr \times \OMN}$ be the kriging operator and $\XSR$ the random image following the distribution of the SR samples generated with Equation~\ref{eq:sample_kriging} then
\begin{equation}
\begin{aligned}
\E_{\XSR}\left(\|\UHR - \XSR\|_2^2\right)
& = \|\UHR-\La^T\ULR\|_2^2 + \Tr\left[(\I_{\OMN}-\La^T\A)\G(\I_{\OMN}-\La^T\A)^T\right] \\
& \geq \|\UHR-\La^T\ULR\|_2^2.
\end{aligned}
\end{equation}
Simply put, the expected mean square error between the optimal HR image and Gaussian SR samples is always higher than the MSE between $\UHR$ and the associated kriging component $\La^T\ULR$.
\end{prop}

Proposition~\ref{prop:MSE}, proved in \ref{appendix:proof_prop_MSE}, means that adding the high frequency content of the innovation component is penalized by the PSNR since it is not aligned with the original high frequency content. Yet this addition is obviously perceptually important. We can also observe that SSIM has a similar behaviour as PSNR, as already observed for other degradation problems \cite{SSIMvsPSNR_Hore_ICPR}.
We argue that our study of Gaussian super resolution shows that LPIPS is the best metric to study the performance of stochastic SR algorithms.

Surprisingly, the three deep learning-based approaches give relatively poor results when applied to  simple textures with high frequency content. 
This suggests that Imagenet and/or DIV2K datasets should probably be enriched with texture images to adapt the models for texture SR.
Given its simplicity, computational efficiency, and qualitative superiority in comparison with state of the art SR methods, we believe our Gaussian SR algorithm is of interest for practitioners interested in stationary texture SR.

\section{Limitations and extensions}
\label{sec:limitations_and_extensions}

As shown in the previous section, when applied to Gaussian textures our Gaussian SR algorithm outperforms state-of-the-art methods regarding results quality and execution time.
However, Gaussian SR has several inherent limitations that we now highlight.

\subsection{Limited scope}

Gaussian microtextures represent a limited scope of images. The images are supposed to be stationary with no geometric structures, excluding a lot of textures such as the brick wall in Figure~\ref{fig:sr_gaussian_fails}. 
When applying our Gaussian SR algorithm to such a structured texture, 
one can observe that the independent stationary stochastic innovation component is not aligned with the brick lines recovered in the kriging component. Also, textures with too big individual objects are unadapted to the Gaussian microtexture assumption, such as the bark texture in the second row of Figure~\ref{fig:sr_gaussian_fails}. The kriging component is attached to the LR data but the innovation component provides a stationary grainy texture that is not in accordance with the piecewise flat nature of the texture. This is a well-known limitation of the ADSN model \cite{Galerne_Gousseau_Morel_random_phase_textures_2011}.

\newlength{\failwidth}
\setlength{\failwidth}{0.196\textwidth}

\begin{figure}
    \centering
    \scriptsize
\begin{tabular}{@{}ccccc@{}}
LR image & HR image & SR Gaussian sample & Kriging comp. & Innovation comp. \\
$\ULR$ & $\UHR$ & $\USR$ & $\La^T\ULR$ & $\tU-\La^T\A\tU$ \\
\includegraphics[width = \failwidth]{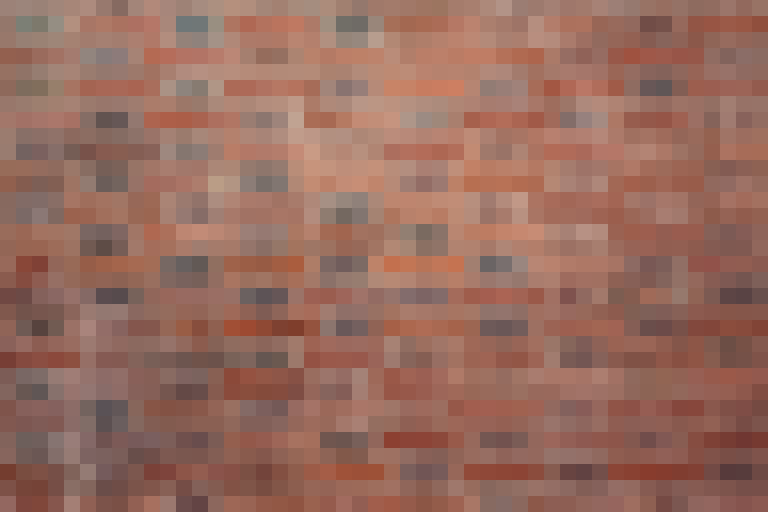}
& \includegraphics[width = \failwidth]{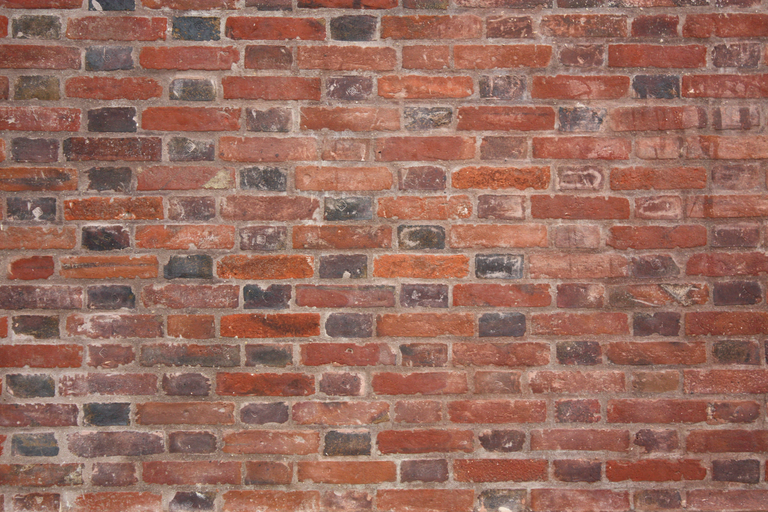}
& \includegraphics[width = \failwidth]{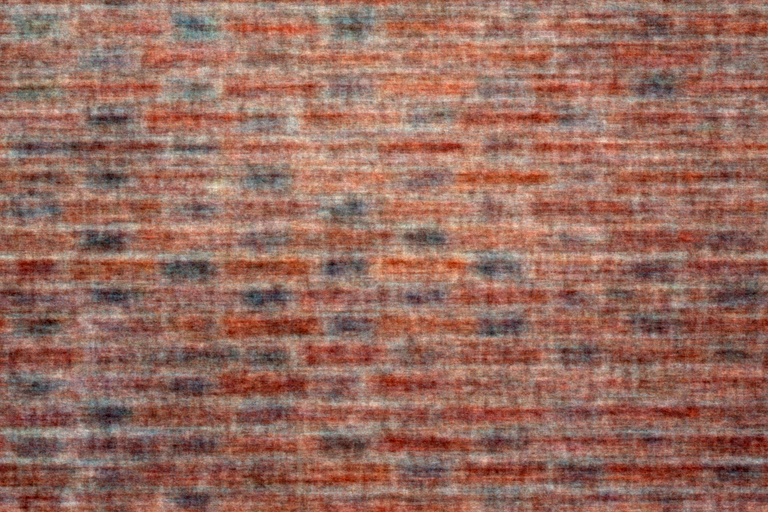}
& \includegraphics[width = \failwidth]{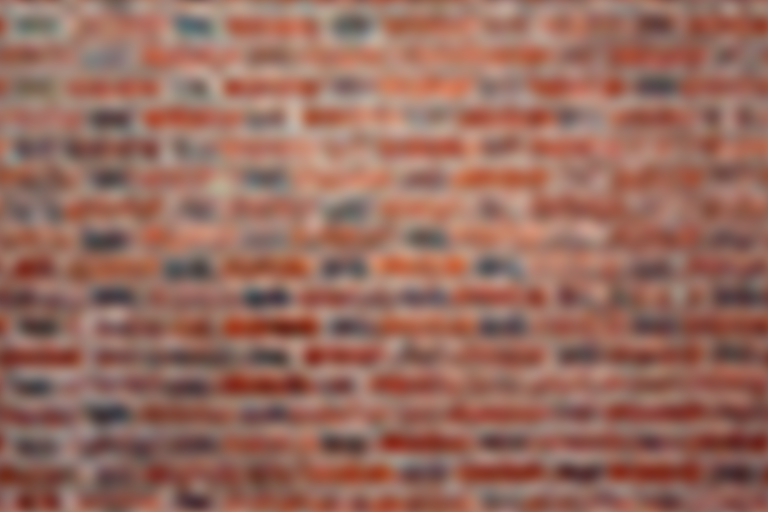}
& \includegraphics[width = \failwidth]{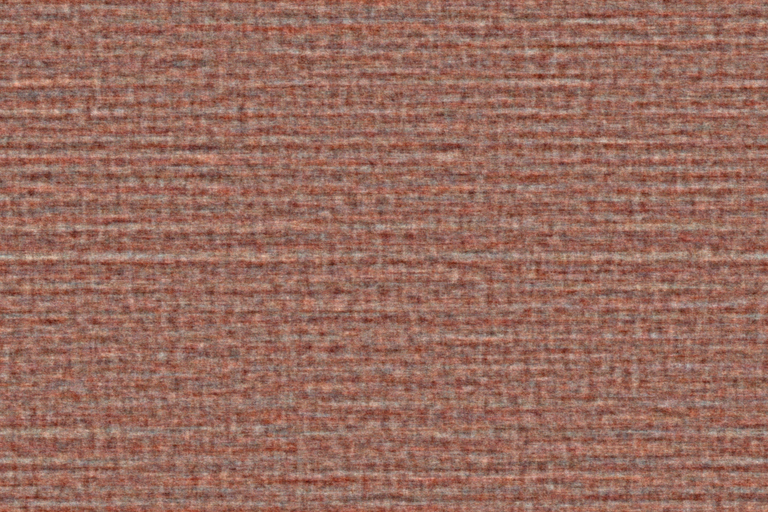} \\
\includegraphics[width = \failwidth]{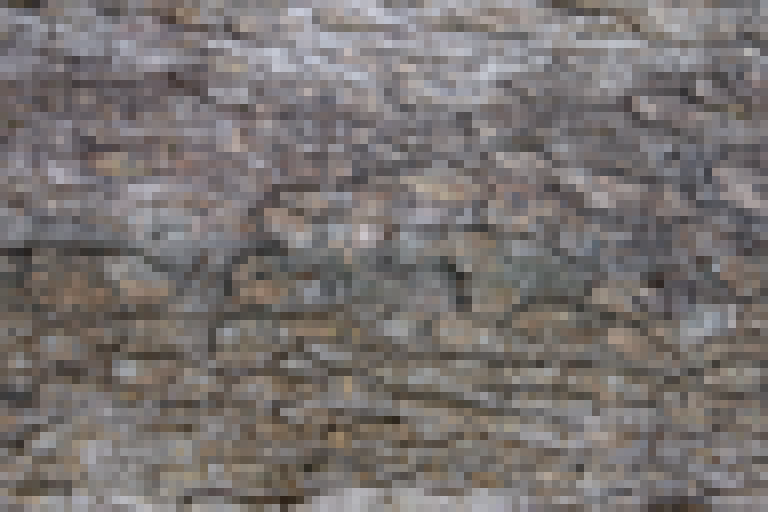}
& \includegraphics[width = \failwidth]{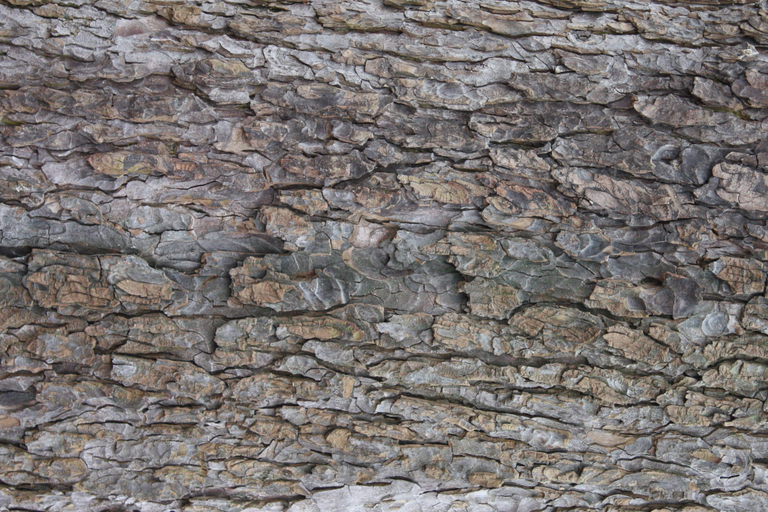}
& \includegraphics[width = \failwidth]{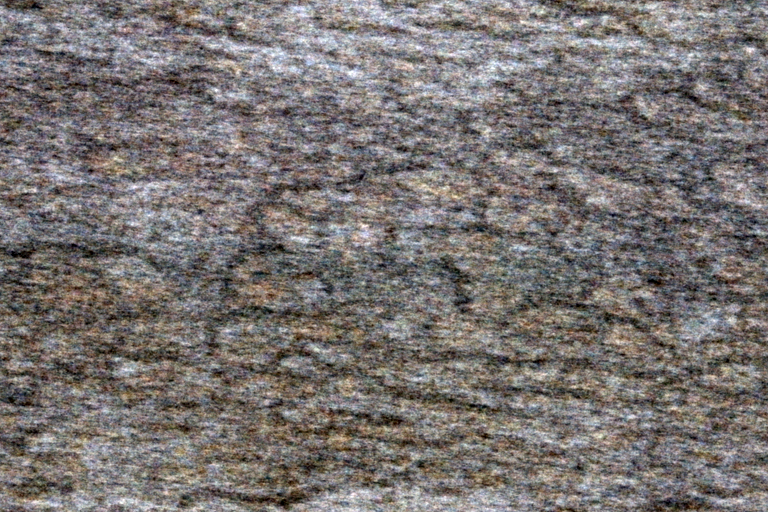}
& \includegraphics[width = \failwidth]{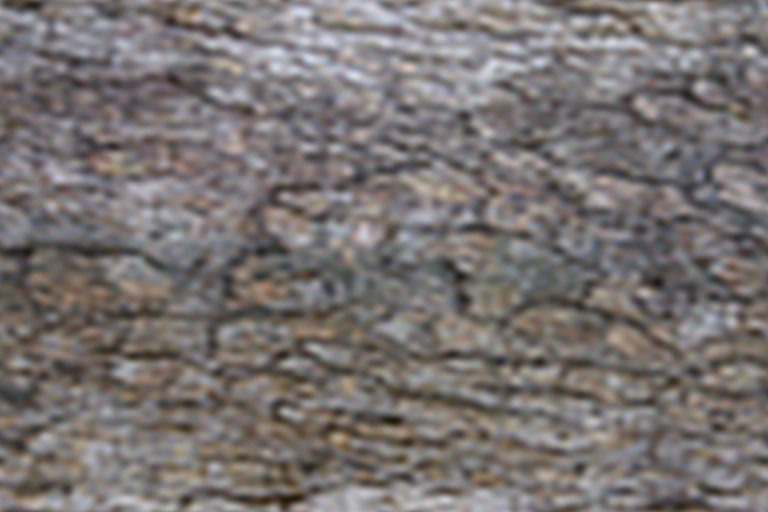}
& \includegraphics[width = \failwidth]{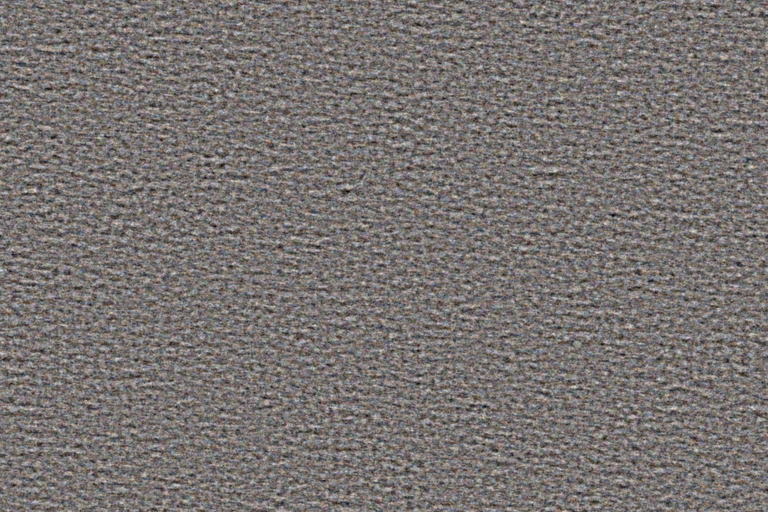}
\end{tabular}
    \caption{\small \label{fig:sr_gaussian_fails}
    Examples of the application of Gaussian SR on images that are not Gaussian microtextures (The HR size is 512$\times$768 and $r = 16$). The first image has a too structured texture to be treated as a Gaussian texture. The second image cannot be considered as a microtexture due to the presence of too big individual patterns.
    }
\end{figure}

\subsection{Non-adapted reference image}

The choice of the reference image has a strong impact on the final samples. The texture is essentially provided by the innovation component which is built with $\ADSN(\Uref)$. Figure~\ref{fig:t_ref} presents two different samples we can obtain depending on the choice of the reference image. The first reference image is the perfect HR image that retrieves the correct alignment of the texture lines. On the contrary, a reference image with opposite alignment modifies the texture of the Gaussian samples.

More generally, with a reference image, the stability result of Proposition~\ref{prop:inequality_lambda} does not hold for the simulation of the kriging component. Algorithm~\ref{algo:sampling_super_resolution_grayscale} can lead to instabilities as illustrated in Figure~\ref{fig:sr_gaussian_fails_non_adapted_ref}. Here we observe that a frequency has been intensified in the kriging component and provokes artificial parallel lines in the final sample. It is essentially due to a slight misalignment between the HR texture and the reference texture. Note that there is no artefact in the innovation component which fits the assumption of Proposition~\ref{prop:inequality_lambda}. This type of observations is marginal but can occur in some rare cases.

\begin{figure}
\scriptsize
\centering
	\begin{tabular}{@{}cccccc@{}}
		LR image & HR image & Reference image  & Sample & Kriging & Innovation \\
		\includegraphics[width =\exampleswidth]{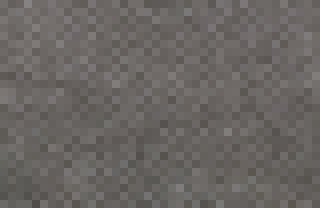}
		& \includegraphics[width = \exampleswidth]{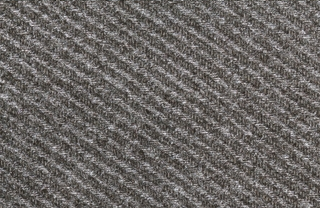}
		& \includegraphics[width = \exampleswidth]{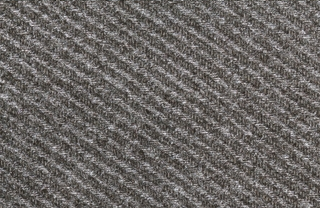}
		& \includegraphics[width = \exampleswidth]{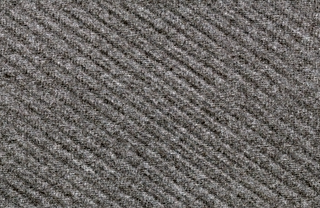}
		& \includegraphics[width = \exampleswidth]{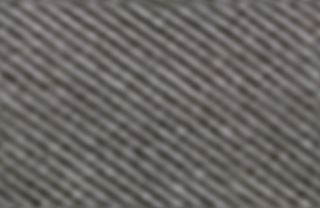}
		& \includegraphics[width = \exampleswidth]{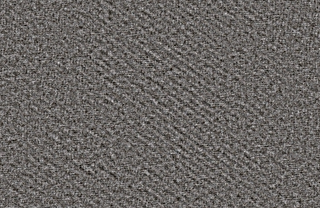} \\
				
		& 
		& \includegraphics[width = \exampleswidth]{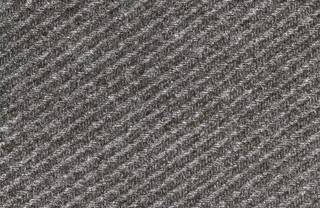}
		& \includegraphics[width = \exampleswidth]{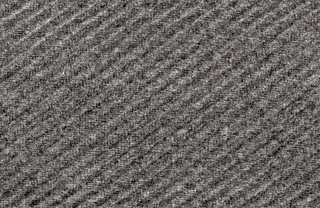}
		& \includegraphics[width = \exampleswidth]{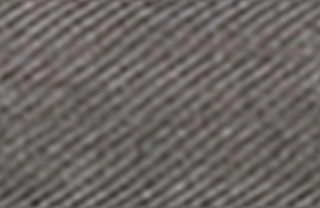}
		& \includegraphics[width = \exampleswidth]{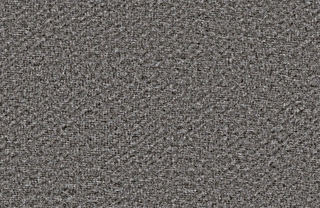} \\

	\end{tabular}
	\caption{\small \label{fig:t_ref} Illustration of the importance of the choice of the reference image. Each row presents the samples of Gaussian SR obtained with the same LR image and a different reference image (The HR size is 208$\times$320 and $r=8$.). The covariance information provided by the reference image guides the final texture sample. }
\end{figure}

\begin{figure}
    \centering
    \scriptsize
\begin{tabular}{@{}cccccc@{}}
\scriptsize
LR image & HR image & Reference image & Gaussian sample & Kriging comp. & Innovation comp. \\
$\ULR$ & $\UHR$ & $\Uref$ & $\USR$ & $\La^T\ULR$ & $\tU-\La^T\A\tU$ \\
\includegraphics[width = \exampleswidth]{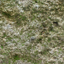}
& \includegraphics[width = \exampleswidth]{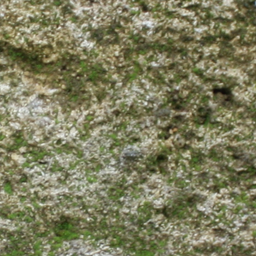}
& \includegraphics[width = \exampleswidth]{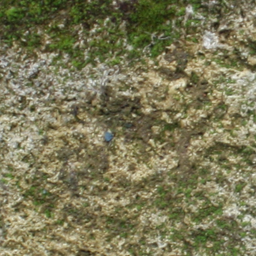}
& \includegraphics[width = \exampleswidth]{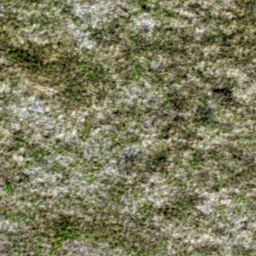}
& \includegraphics[width =\exampleswidth]{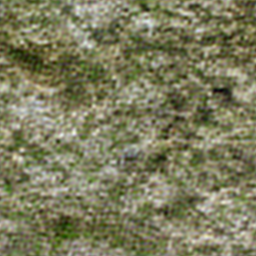}
& \includegraphics[width = \exampleswidth]{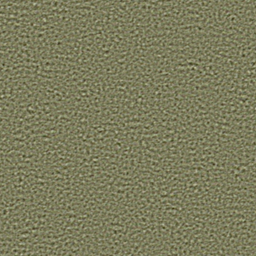} \\
\end{tabular}
    \caption{\small \label{fig:sr_gaussian_fails_non_adapted_ref}
    Example of Gaussian SR that illustrates the non-application of Proposition~\ref{prop:inequality_lambda} in the setting with reference image (HR size is 256$\times$256 and $r = 4$). A frequency has been amplified and provokes parallel lines in the kriging component generated by Gaussian SR. This artefact occurs in rare cases. This phenomenon is not observable in the innovation component which verifies the assumptions of Proposition~\ref{prop:inequality_lambda}.
    }
\end{figure}

\subsection{Non-adaptative variance}

We can study the variance of the samples generated by our method with Proposition~\ref{prop:var}. As a reminder, in our approximation $\La_{\text{approx}}^T$ is in the form $\G\A^T\bC_{\etab}$ where $\bC_{\etab}$ is a diagonal multi-chanel convolution but this is also the case with the perfect solution of Equation~\ref{eq:kriging_matrix} $\La = \G\A^T\left(\A\G\A^T\right)^\dagger$. This proposition shows that the images sampled with Gaussian SR have a constant variance on sub-grids of the image domain. The law of the SR samples $\USR$ inherits the stationarity of $\ADSN(\Uref)$ being invariant by the translations by $(kr,\ell r)$ where $k,\ell \in \Z$. In other words, the law of our SR samples is a cyclostationary law \cite{lutz_2021_cyclostationary_computer_graphics_forum}.

\begin{prop}[Non-adaptative variance of Gaussian SR]
\label{prop:var}
	Let $\La^T \in \R^{3\OMN \times 3\OMNr}$ being in the form $\bC_{\etab}\Sub^T$ where the kernels $\etab\in \R^{3\OMN}$.
	Then, the law of the SR samples generated by Equation~\ref{eq:sample_kriging} is invariant by translations by $(kr,\ell r)$ where $k,\ell \in \Z$.
\end{prop}

\noindent Proposition~\ref{prop:var} implies that the samples generated by Gaussian SR have a pixelwise variance constant by translations $(kr,\ell r)$ where $k,\ell \in \Z$. This variance can be computed theoretically and compared to the empirical pixelwise variance of the other stochastic algorithms, as done in Figure~\ref{fig:var}. The white spot detail in the textures is taking into in account in the variety of the samples proposed by the other stochastic super-resolution algorithms but not by our Gaussian SR. The variance of SRFlow is quite high and allow us to see the input image, it is clearly adapted to the LR image. DPS variance map show the white spot detail but one can guess that the texture is not really retrieved because the variance is smooth on the rest of the image. DDRM is quasi deterministic with a very low variance value. Gaussian SR presents a constant variance on sub-grids of the image.  
Indeed, one of the main limitation of the model is that the innovation component is independent of the kriging component.

\newlength{\variancewidth}
\setlength{\variancewidth}{0.196\textwidth}

\newlength{\barheight}
\setlength{\barheight}{3.2cm}

\newlength{\barwidth}
\setlength{\barwidth}{.95\variancewidth}
\begin{figure}
\centering
\scriptsize
\begin{tabular}{@{}ccccc@{}}
	HR image & Gaussian SR  & SRFlow  & DPS & DDRM \\
	\includegraphics[width=\variancewidth]{Comparison_Wall/HR}
    & \includegraphics[width=\variancewidth]{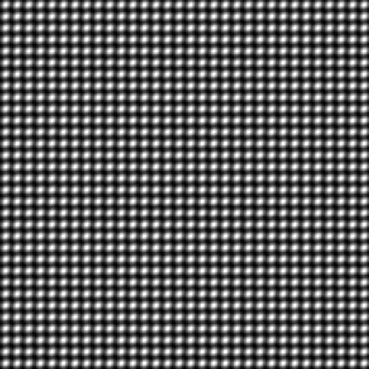}
	& \includegraphics[width=\variancewidth]{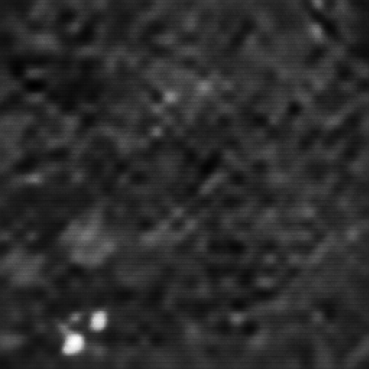}
	& \includegraphics[width=\variancewidth]{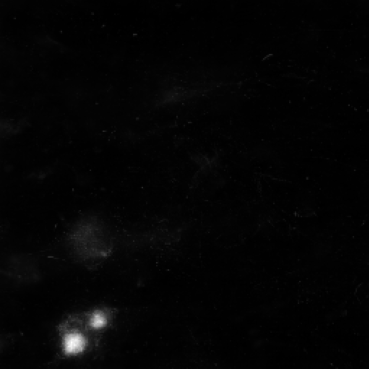} 
	& \includegraphics[width=\variancewidth]{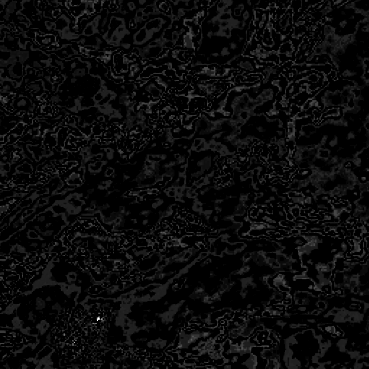}\\
	&  \includegraphics[width=\barwidth]{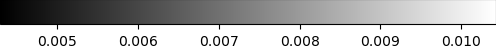} 
		& \includegraphics[width=\barwidth]{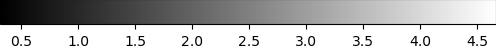}
	& \includegraphics[width=\barwidth]{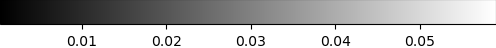}
	& \includegraphics[width=\barwidth]{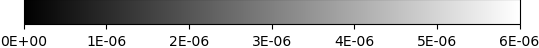}
\end{tabular}
	\caption{\small \label{fig:var} Variance maps of the different algorithms for stochastic super-resolution applied to the HR image of Figure~\ref{fig:comparaison_methods_Wall}, with $r = 8$. The variance of Gaussian SR has been computed theoretically while it has been estimated with $100$ samples for the others, and the sum of the pixelwise variance over the three chanels is displayed. While DDRM is close to being deterministic, the SRFlow and DPS algorithms have an adaptative variance to the HR image, showing the white spot of the image. In comparison, the  Gaussian SR has a cyclostationary non-adaptative variance.}
\end{figure}

\subsection{Extension to other degradation operators}

Our method can be applied to reconstruct images degraded by other degradation operators in the form $\Sub\bC$ where $\bC$ is a convolution preserving the mean. All the proofs and the properties still hold. To illustrate it, we consider motion blurs randomly generated by the code\footnote{\url{https://github.com/LeviBorodenko/motionblur}} with kernel size $61 \times 61$  and intensity value $0.5$, following \cite{Chung_DPS_ICLR_2023}. We propose to reconstruct images degraded by a motion blur followed by a subsampling operator $\Sub$ with stride $r=4$ and images degraded by a motion blur and the bicubic zoom-out operator $\A$ with factor $r=4$. These are not standard problems but they can be seen as multiple problems as studied in \cite{song_pseudoinverse_guided_diff_models_2023_ICLR}. In Figure~\ref{fig:motion_blur}, we present the results provided by our method. We can again observe the complementarity of the deterministic kriging component and the stochastic innovation component that retrieves the perceptual texture grain.
\newlength{\blurwidth}
\setlength{\blurwidth}{0.139\textwidth}

\begin{figure}
\centering 
\tiny
\begin{tabular}{@{}ccccccc@{}}
     \multicolumn{7}{c}{Motion blur followed by a subsampling operator with stride $r=4$ }\\
    \midrule
	Kernel & Degraded image & Ground truth & Reference & Kriging comp. & Innovation comp. & Sample \\
		\includegraphics[width=\blurwidth]{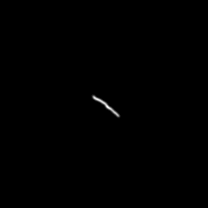} &
	\includegraphics[width=\blurwidth]{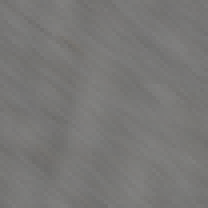} &
	\includegraphics[width=\blurwidth]{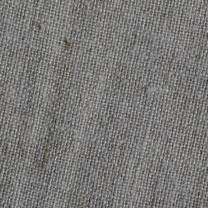} &
	\includegraphics[width=\blurwidth]{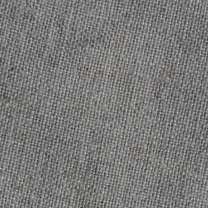} &
	\includegraphics[width=\blurwidth]{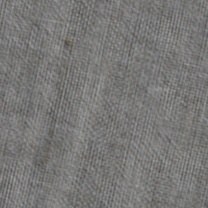} &
	\includegraphics[width=\blurwidth]{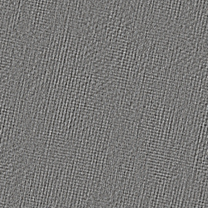} &
	\includegraphics[width=\blurwidth]{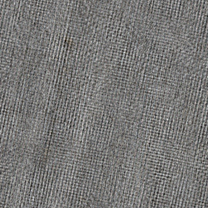} \\
	\includegraphics[width=\blurwidth]{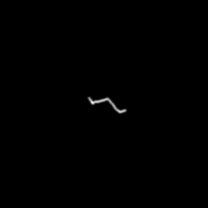} &
	\includegraphics[width=\blurwidth]{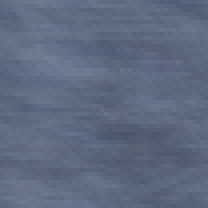} &
	\includegraphics[width=\blurwidth]{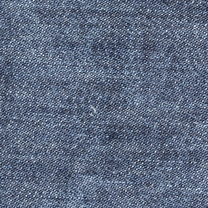} &
	\includegraphics[width=\blurwidth]{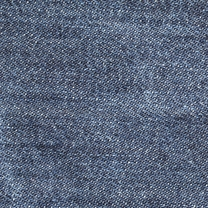} &
	\includegraphics[width=\blurwidth]{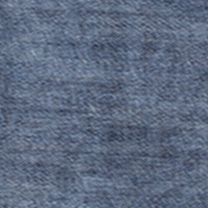} &
	\includegraphics[width=\blurwidth]{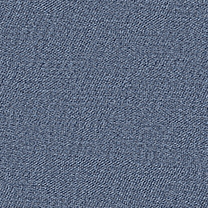} &
	\includegraphics[width=\blurwidth]{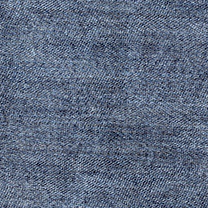} \\
	 \multicolumn{7}{c}{Motion blur followed by a bicubic convolution and a subsampling operator with stride $r=4$} \\
    \midrule
		\includegraphics[width=\blurwidth]{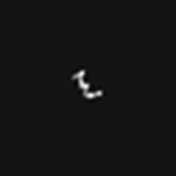} &
	\includegraphics[width=\blurwidth]{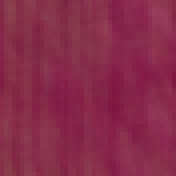} &
	\includegraphics[width=\blurwidth]{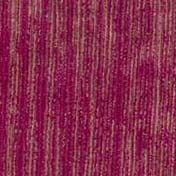} &
	\includegraphics[width=\blurwidth]{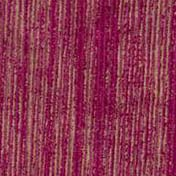} &
	\includegraphics[width=\blurwidth]{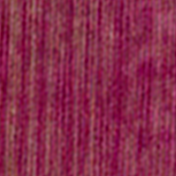} &
	\includegraphics[width=\blurwidth]{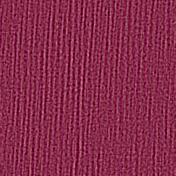} &
	\includegraphics[width=\blurwidth]{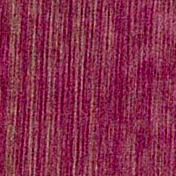} \\
			\includegraphics[width=\blurwidth]{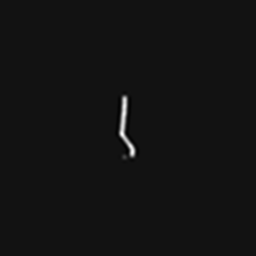} &
	\includegraphics[width=\blurwidth]{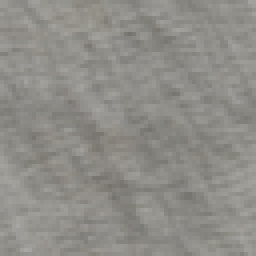} &
	\includegraphics[width=\blurwidth]{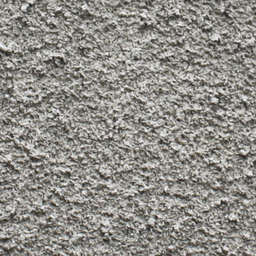} &
	\includegraphics[width=\blurwidth]{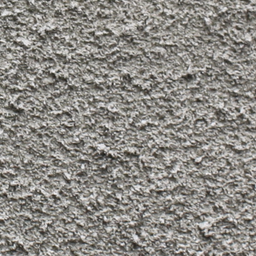} &
	\includegraphics[width=\blurwidth]{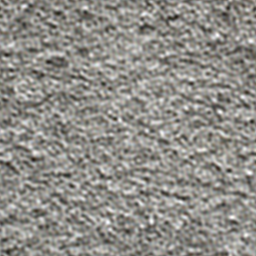} &
	\includegraphics[width=\blurwidth]{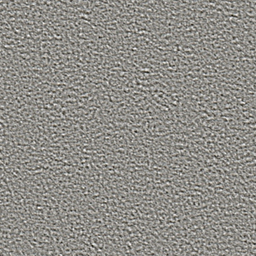} &
	\includegraphics[width=\blurwidth]{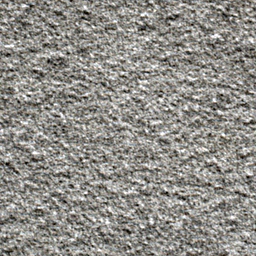} \\
	\end{tabular}
\caption{\small \label{fig:motion_blur} Illustration of the application of Gaussian kriging to other linear operators. The images have square size which are respectively $208, 208, 176$ and $256$. The two first images are degraded by a motion blur, followed by a subsampling with stride $r=4$ and the two others are degraded by a motion blur followed by the zoom-out bicubic operator with factor $r=4$.}
\end{figure}

\section{Conclusion}
\label{sec:conclusion}

This work provides an efficient sampling algorithm for solutions of the stochastic SR problem for Gaussian microtextures.
Our solution is exact for grayscale images and a we propose a fast and reliable approximation for RGB images using a well-known mathematical kriging reasoning.
Gaussian SR is a fast and qualitatively better than involved generic methods when applied to Gaussian microtextures accompanied with a reference HR image.
Besides, of a more general interest, our comparative study shows that LPIPS is a pertinent metric to evaluate stochastic SR results that do not suffer from excessive blur.
While our Gaussian SR has inherent limitations discussed in the paper, 
it could be of interest for fast and reliable restoration of HR stationary microtextures, even in the presence of motion blur.

\bibliographystyle{plain}
\bibliography{refs.bib}

\appendix

\section{Proof of Theorem~\ref{thm:kriging}}

\label{appendix:proof_kriging_equation}

\noindent The fact that the conditional expectation is linear for zero-mean Gaussian multivariate law is a classical result (see e.g. \cite{doob_stochastic_1953}). Let us denote $(\la(x))_{x \in \OMN}$ the columns of $\La$.
$\La$ minimizes
$$
\begin{aligned}
    \E\left(\left\|\La^T\A\X -\X\right\|_2^2\right)
    & = \E\left[\sum_{x \in \OMN}\left((\La^T \A\X)(x)-\X(x)\right)^2\right] \\
    & = \sum_{x \in  \OMN}\E\left[\left(\la(x)^T\A\X-\X(x)\right)^2\right].
\end{aligned}
$$
This is a separable function of the columns $(\la(x))_{x \in \OMN}$. 
For each for $x \in  \OMN$, let us minimize 
$\la(x)~\in~\R^{\OMNr}~\mapsto~\E\left[\left(\la(x)^T\A\X-\X(x)\right)^2\right]$.
One has
$$
\begin{aligned}
    \E\left[\left(\la(x)^T\A\X-\X(x)\right)^2\right]
    & = \Var(\la(x)^T \A \X) + \Var(\X(x)) -2 \Cov(\la(x)^T \A\X,\X(x)) \\
    & =\la(x)^T \A\G \A^T \la(x) + \Var(\X(x)) -2\la(x)^T\A \Cov(\X,\X(x)).
\end{aligned}
$$
This is a quadratic functional associated with the positive matrix $2 \A\G \A^T$ and the vector $2\A\Cov(\X,\X(x)) = 2 \A\G_{ \OMN \times \{x\}}$. It is minimal for $\la(x)$ any solution of the linear system $\A\G \A^T\la(x) = \A\G_{\OMN \times \{x\}}$. This is valid for any column $\la(x)$ of $\La$. As a consequence, $\La$ is a solution if and only if it verifies
$$
\A\G \A^T\La = \A\G.
$$

\section{Proof of Lemma~\ref{lem:convolution_subsampling}}

\label{appendix:proof_convolution_subsampling}

\noindent  It is sufficient to identify  the convolutions on the canonical basis. For $\x \in \OMN$, we define $\delta_\x \in \R^{\OMN}$ such that for $\z \in \OMN$, $\delta_\x(\z) = \mathbb{1}_{\x = \z}$ and identically $\delta^{r}_\x$ for $\x \in \OMNr$. Denoting $\Tx$ the translation by $\x\in \OMN$, that is for all $\z\in\R^{\OMN}$, $\T_\x \U(\z) = \U(\z-\x)$. Remark that convolving with $\delta^{\OMN}_\x$ corresponds to translating by $\x$: for all $\U \in \R^{\OMN}, \U \star \delta_\x = \T_{\x}\U$.

 \begin{enumerate}
	\item For $\x \in \OMNr$, 
\begin{equation}
	\Sub\bC_{\balpha}\Sub^T\delta^{ r}_\x 
	= \Sub\bC_{\balpha} \delta_{r\x}  
	= \Sub (\balpha \star \delta_{r\x}) 
	= \Sub \T_{r\x} \balpha 
	= \Tx \Sub \balpha  
	= (\Sub \balpha) \star \delta^{r}_\x.
\end{equation}
	\item For $\x \in \OMNr$,
\begin{equation}
	\Sub^T\bC_{\be}\delta^{r}_\x  
	= \Sub^T (\be \star \delta^{r}_{\x})
	= \Sub^T \T_{x}\be 
     =  \T_{r\x} \Sub^T \be 
     = (\Sub^T \be) \star \delta_{rx} 
     = (\Sub^T \be) \star( \Sub^T \delta^{r}_\x ).
\end{equation}
\end{enumerate}
 
\noindent Finally, the equality of the convolutions on the canonical basis is established.
	
\section{Conjugate Gradient Descent algorithm}
\label{appendix:algo_CGD}

\noindent The CGD algorithm is an iterative method to approximate the solution of a linear equation by solving the associated least-squares problem \cite{Kammerer_convergence_CGD_1972_SIAM}. Algorithm~\ref{algo:CGD} aims at solving the normal equation
\begin{equation}
	\B^T\B\psi = \B^T\varphi 
\end{equation}
associated with the least-squares problem $\psi \mapsto \left\|\B\psi- \varphi\right\|_2$ and has good convergence properties when applied to the inpainting problem for Gaussian textures \cite{Galerne_Leclaire_gaussian_inpainting_siims2017}.
First we observed that this also applies in the case of super-resolution.
In our SR context, our goal is to express $\left(\A\G\A^T\right)^\dagger \V$ for a given LR image $\V \in \R^{\OMNr}$. It is equivalent to minimize $\psi \in \R^{\OMNr} \mapsto \left\| (\A\G\A)^T \psi - \V \right\|_2$. Therefore, we apply CGD with $B = \A\G\A^T$ and $\varphi = \V$.

\begin{algorithm}[H]
\caption{\small CGD Algorithm CGD to compute $\B^\dagger \varphi$}
\label{algo:CGD}
\begin{algorithmic}
\STATE \textbf{Input:} Initialize $k \leftarrow 0, \psi_0 \leftarrow 0, r_0 \leftarrow \B^T\varphi - \B^T\B\psi_0, d_0 \leftarrow r_0$
\WHILE{$\|r_k\|_2 > \varepsilon$}
\STATE $\alpha_k \leftarrow \frac{\|r_k\|_2^2}{d_k^T\B^T\B d_k}$
\STATE $\psi_{k+1} \leftarrow \psi_k +\alpha_k d_k$
\STATE $r_{k+1} \leftarrow r_k - \alpha_k \B^T\B d_k$
\STATE $d_{k+1} \leftarrow r_{k+1} + \frac{\|r_{k+1}\|_2^2}{\|r_{k}\|_2^2}d_k$
\STATE $k \leftarrow k+1$
\ENDWHILE
\STATE \textbf{Output:} $\psi_k$
\end{algorithmic}
\end{algorithm}

\section{Proof of Proposition~\ref{prop:MSE}}
\label{appendix:proof_prop_MSE}

\noindent Let $\UHR \in \R^{\OMN}$ be a HR image, $\ULR = \A\UHR$ its LR version, $\Uref$ a reference image,  $\G$ such that $\ADSN(\Uref) = \mathscr{N}(\zero,\G)$, $\La \in \R^{\OMNr \times \OMN}$ be the kriging operator and $\XSR$ the random image following the distribution of the SR samples generated with Equation~\ref{eq:sample_kriging}. $\XSR$ has the same law as:
$$\USR = \La^T \ULR + \tX - \La^T\A\tX$$
\noindent with $\tX \sim \mathscr{N}(\zero,\G)$. Consequently, considering $\E_{\tX}(\tX) = \zero$,
	$$
	\begin{aligned}
		\E_{\XSR}\left(\|\UHR - \XSR\|_2^2\right)
		= & \|\UHR-\La^T\ULR\|_2^2+\E_{\tX}\left(\|\tX - \La^T\A\tX\|_2^2\right) \\
		& +2\langle \UHR-\La^T\ULR,\E_{\tX}( \tX) -\La^T\A \E_{\tX}(\tX)\rangle  \\
		= & \|\UHR-\La^T\ULR\|_2^2+\Tr\left[(\I_{\OMN}-\La^T\A)\G(\I_{\OMN}-\La^T\A)^T\right]\\
	\end{aligned}
	$$
	
\noindent	And, $(\I_{\OMN}-\La^T\A)\G(\I_{\OMN}-\La^T\A)^T$ is a positive semi-definite matrix with non-negative trace which gives
	\begin{equation}
		\E_{\XSR}\left(\|\UHR - \XSR\|_2^2\right)\geq  \|\UHR-\La^T\ULR\|_2^2.
	\end{equation}
	
	\section{Proof of Proposition~\ref{prop:var}}
	\label{appendix:proof_prop_var}

\noindent	Let $\La^T \in \R^{3\OMN \times 3\OMNr}$ being in the form $\bC_{\etab}\Sub^T$ where the kernel $\etab\in \R^{3\OMN}$.
	The law of $\tX - \La^T(\A\tX)$ is invariant by translation by $(rk,r\ell)$ for $k,\ell \in \Z$ since $\tX$ follows a stationary law, $\La^T\A\tX = \bC_{\etab}\star \Sub^T(\A\tX)$ with $\A\tX \in \R^{3\OMNr}$ following a stationary law.

\end{document}